\documentclass{article}

\usepackage[utf8]{inputenc}
\usepackage{ifpdf}
\usepackage{eurosym}
\usepackage{amsmath}
\usepackage{amsfonts}
\usepackage{amsthm}
\usepackage{makeidx}
\usepackage{amssymb}
\usepackage{graphicx}
\usepackage{epsfig}
\usepackage{color}

\usepackage{hyperref}

\theoremstyle{plain}
\newtheorem{lemma}{Lemma}[section]

\newtheorem{algorithm}[lemma]{Algorithm}

\theoremstyle{definition}
\newtheorem{remark}[lemma]{Remark}
\newtheorem{example}[lemma]{Example}
\newtheorem{definition}[lemma]{Definition}

\newcounter{example}

\newcommand{\E}{{\mathbb E}}
\newcommand{\N}{{\mathbb N}}
\renewcommand{\P}{{\mathbb P}}
\newcommand{\R}{{\mathbb R}}

\newcommand{\Z}{{\mathbb Z}}

\newcommand{\cF}{{\cal F}}

\newcommand{\diag}{\mathrm{diag}}

\newcommand{\icomp}{\mathrm i}

\newcommand{\notion}[1]{{\em #1}\index{#1}}
\newcommand{\exclude}[1]{}

\title{A short introduction to quasi-Monte Carlo option pricing}

\author{Gunther Leobacher\thanks{The author is supported by the Austrian 
Science Fund (FWF) project F5508-N26, which is part of the Special
Research Programme ``Quasi-Monte Carlo Methods: Theory and Applications''}}

\date{2014}

\begin{document}

\maketitle 

\begin{abstract}
One of the main practical applications of quasi-Monte Carlo (QMC) methods
is the valuation of financial derivatives. We aim to give a short introduction
into option pricing and show how it is facilitated using QMC. 
We give some practical examples for illustration.
\end{abstract}

%\tableofcontents

\section{Overview}

Financial mathematics, and in particular option
pricing, has become one of the main application of quasi-Monte Carlo (QMC)
methods. By QMC we mean the numerical approximation of high-dimensional
integrals over the unit cube
\[
I=\int_{[0,1]^d}f(x) dx
\]
by deterministic equal weight integration rules, that is
\[
I\approx\frac{1}{N}\sum_{k=0}^{N-1}f(x_k)\,,
\]
for a suitably chosen point set $x_0,\ldots,x_{N-1}\in [0,1]^d$.
\exclude{Classically, 
$D=[0,1)^d$ and $\lambda$ is Lebesgue measure. For financial applications,
in particular for models relying on Gaussian random variables, it is often more
convenient to consider the setup $D=\R^d$, $\mu(A)=\int_A \phi(x)dx$, where
$\phi(x)=\exp(-\frac{x^\top x}{2})\frac{1}{\sqrt{2\pi}^d}$.}\\

In Section \ref{sec:finance} we give a very brief introduction into
the theory of option pricing. The main intention is to explain why 
an option price can be written (approximately!) as a high dimensional integral.
We present a couple of examples which are frequently used by researchers as
benchmarks for their pricing methods.\\

In Section \ref{sec:simulation} we first discuss some generalities of
simulation, like the generation of non-uniform random variables. We 
give some arguments why acceptance-rejection algorithms usually do not 
work so well with QMC. We give the basic properties of Brownian motion
and of L\'evy processes and we show how approximate paths can be 
generated from uniform or normal input variables. A special emphasis
is on orthogonal transforms for path generation. We mention the important
topic of multilevel Monte Carlo and we conclude with some concrete 
examples from option pricing.\\

This article does not try to be a comprehensive survey. There are many problems
and solutions that do not find any mention here but which are no less
important.  Just to mention one topic: for barrier options the 
discretization bias, when using the maximum of a discrete Brownian path as an
approximation to the continuous time path, is very big and
thus leads to impractically high dimensions. Therefore one has to find
ways to sample from the maximum of the path between discretization nodes
or similar, thus using
more involved probability theory than is required to understand the basic
methods presented here.

This article is also not comprehensive in that it neglects an important point:
Why do these methods work for financial problems? Most of the theory of QMC
does not apply to the kinds of functions appearing in option pricing. These
function are usually well behaved in that they are piecewise log-linear, but
they are very high dimensional, they are in general not bounded or of 
bounded variation, nor do they lie in any of
the many weighted Korobov or Sobolev spaces for which integration has been
proven to be tractable. 
Nevertheless the methods described in this articles are widely used in practice
and they do seem to work quite well. To fully explain why they give
these good results would be a great achievement and is subject to active
research.

\section{Foundations of Financial Mathematics}\label{sec:finance}

\subsection{Bonds, stocks and derivatives}

Since financial mathematics is (mainly) about the valuation of financial
instruments, we now
give a short overview of the most basic of these. 

\begin{itemize}
\item A {\em bond} is a financial instrument that pays its owner a fixed 
amount of money at a pre-specified date in the future. The writer of the bond
is usually a big company or a government. The owner effectively becomes a
creditor to the writer. If the quality of the debtor is high, the bond
can be modeled as a deterministic payment. The bond usually sells at a lower
price than its payoff and thus pays {\em interest}.
\item A {\em share} is a financial instrument that warrants its holder
ownership of a fraction of a corporation. In particular, the shareholder
participates in the business revenue due to dividend payments.

However, dividend payments are not the only possible source of income
through a share. At least equally important is the gain due to a price change.
On the downside the price change may result in a loss.
If the shares of a company are traded at a stock exchange then buying and 
selling them is particularly simple and high frequency traders
may buy and sell large contingents of shares several times per second.

The value of a share depends on a host of parameters, such as the preferences 
of the individual agent, the assets of the company, the future 
dividend payments and the future interest rates. 

The so-called {\em efficient market hypothesis} assumes that the value of 
the share at a given time is just the market price at that very time.
Under this hypothesis it does not make sense to compute the objective 
value of a share in a mathematical model and compare it to the market price.
The only way that a computed value of a share can differ from its market 
price is that our preferences and/or expectations differ from that of the 
majority of the market, thus giving a subjective price.

\item  A {\em contingent claim} is a financial instrument whose value at a
future date can be completely described in terms of 
the prices of other financial instruments, its underlyings. A typical example
is an option on a share. A European call option on a share with maturity
$T$ is a contract which gives its holder the right (but not the obligation) 
to buy one share from the option writer at some fixed time $T$ in the future
at the previously 
agreed price $K$. Denote the price of the share at time $T$ by $S_T$.

Since the option holder may sell the share instantly at
the stock exchange, the value of the option at time $T$ 
is $S_T-K$ if $S_T>K$, and $0$ if $S_T\le K$. 

The left hand side of figure \ref{fig:payoff-vanilla} shows the payoff of a 
European call option dependent on the price of the share at maturity.
An important feature is the kink at the strike price $K$. \exclude{If the 
price depends on several input variables, this kink will have the effect
that the concatenated function has unbounded variation.}

On the right hand side of figure \ref{fig:payoff-vanilla} we plot the 
payoff of a European put option. This is an option which gives its holder
the right (but not the obligation) to sell one share 
to the option writer at some fixed time $T$ in the future
at the previously agreed price $K$. If the share price satisfies $S_T\ge K$ at
time $T$, then the option is worthless. But if $S_T<K$, then the option holder
may buy the
share at the stock exchange at price $S_T$ and sell it immediately to the
writer at price $K$, thus realizing a gain of $K-S_T$.

\begin{figure}[h]
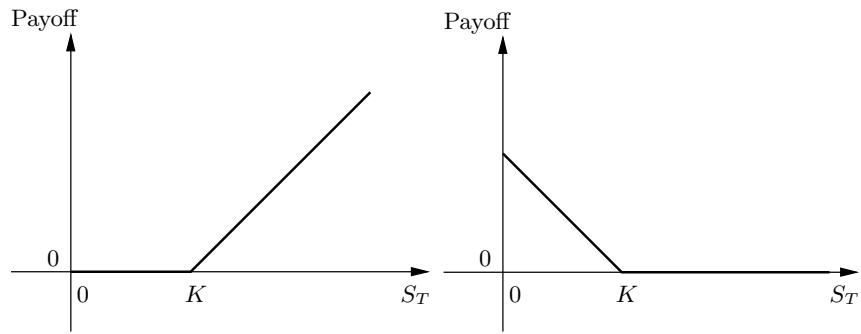

\begin{center}
\input{payoff_call.tex}
\input{payoff_put.tex}
\end{center}
\caption{Payoff of a European call and put option}\label{fig:payoff-vanilla}
\end{figure}

In Figure \ref{fig:payoff-digital} we show the payoff of another contingent
claim, a so-called digital asset-or-nothing call option. This option pays a
fixed amount of cash
at expiry if at that time the price $S_T$ of the underlying  
is above the strike $K$. We also plot the payoff of the corresponding put option. The digital option serves as an example of a contingent claim with 
discontinuous payoff.
\\
\begin{figure}[h]
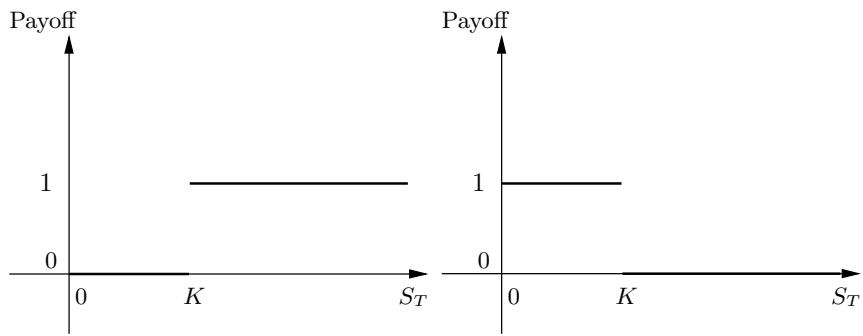

\begin{center}
\input{payoff_digcall.tex}
\input{payoff_digput.tex}
\end{center}
\caption{Payoff of a digital cash-or-nothing call and put option}\label{fig:payoff-digital}
\end{figure}

Since the value of an option 
is strongly tied to that of the underlying and in simple models
is completely determined by the parameters of the model,
an objective value of the option can be computed in these models using
arbitrage arguments.
\end{itemize}

\subsection{Arbitrage and the No-Arbitrage Principle}

Suppose you are given an option on a stock. You know 
the specifications of the option and therefore you know 
the uncertain payoff at its maturity $T$ given the uncertain
value of the stock at that particular time. One is tempted to use
statistical methods to estimate from historical stock prices 
the distribution of the stock price at time $T$ and {\em a fortiori}
estimate the value of the option as its expected value under the
estimated distribution. We will show in this section that this 
reasonable program will in general yield a price that is unreasonable
from a more basic perspective in that it allows for risk-less profit.

While general arbitrage theory is well beyond the scope of this article,
the underlying principle can be illustrated rather quickly. For the general
theory see \cite{ds06}.

Assume the following simple market model where we have only two times, $0$ and
$1$, and three instruments, a bond, a share, and a European call option with
strike $K=1$ and maturity $T=1$.  Let
$B=(B_t)_{t\in\{0,1\}},S=(S_t)_{t\in\{0,1\}},C=(C_t)_{t\in\{0,1\}}$ denote the
price processes of the bond, share, option respectively
and assume the following parameters: $B_0>0$, $B_1=B_0(1+r)$, $r\ge 0$,
$S_0>0$, $S_1=S_0 u$ with probability $p$ and $S_1=S_0 d$ with probability 
$1-p$, where $0<d<1+r<u$. The value of the option at time 1 is 
$\max(S_1-K,0)$, thus
$C_1=\max(S_0u-K,0)$ with probability $p$ and $C_1=\max(S_0d-K,0)$ with 
probability $1-p$.
Suppose we know, for example from statistical studies, the value of  $p$. 

Then one would be tempted to conclude that the price of the option at 
time $0$ is
\[
\hat C_0 = \frac{B_0}{B_1}\E(\max(S_1-K,0))=
\frac{1}{1+r}(p \max(S_0 u-K,0)+(1-p)\max(S_0 d-K,0))\,.
\]
However, this formula cannot be true in general. Suppose $r=0$, $u=2$,
$d=\frac{1}{2}$, $S_0=K=1$ and $p=\frac{1}{2}$, for which the
above formula gives $\hat C_0=\frac{1}{2}$. 

Then we could do the following: at time 0,
write  4 options and sell them for 2 Euros, borrow 1 additional Euro 
to buy three shares. Note that the net
investment is zero.

Now wait until time 1. 
If the share price goes up, the shares are worth 6. We sell them 
to get 6 Euros in Cash. Since the share price $S_T$ (which is $2$) 
is bigger than the strike $K$ (which is $1$), the
options will be executed, costing us 4 Euros and we have to pay 1 Euro back.
Thus our strategy leaves us with a net profit of 1 Euro. 

If the share price goes down, the options become worthless and we sell
the shares, giving us $\frac{3}{2}$ and thereby, after paying 
1 Euro back, leaving us with a profit of $\frac{1}{2}$ Euro.

Thus, whatever happens, we are left with a positive profit without
taking any risk. Such a situation is called an {\em arbitrage opportunity} and
for obvious reasons it is usually assumed that such opportunities
 do not exist in a viable market.

It can easily be shown that there is only one price in this model that does
not allow for arbitrage, namely
\[
C_0=\frac{1}{1+r}\big(p^* \max(S_0 u-K,0)+(1-p^*)\max(S_0 d-K,0)\big)\,,
\]
where $p^*=\frac{1+r-d}{u-d}$. The distinctive feature of $p^*$ is that
$\frac{1}{1+r}(p^*S_0u+(1-p^*)S_0d)=S_0$, that is,  
the stock price process is a {\em martingale}\footnote{We do not give a precise
definition for this. Intuitively, a martingale is a process $X$ such that
the conditional expectation of $X_{t+s}$ given $X_t$ is $X_t$. Thus a 
martingale is a model for the gain process of a player in a fair game.} 
with respect to this new probability.

\subsection{The Black-Scholes model}

The simple model in the preceding section can be extended to 
an $n$-step setup. One is tempted to let $n$ go to infinity to
obtain a continuous-time model. Indeed, this can be done in 
rigorous fashion so that we arrive at a model of
the form 
\begin{equation}\label{eq:bs-model}
\begin{aligned}
B_t&=B_0 \exp(rt)\\
S_t&=S_0 \exp(\mu t + \sigma W_t)\,, 
\end{aligned}
\end{equation}
where $W$ is a Brownian motion, that is, a continuous-time stochastic
process with specific properties.  The exact mathematical definition
of Brownian motion will be given in Section \ref{sec:brownian-motion}.

As in the one-step model there exists a probability measure $\P^*$,
equivalent to the original measure $\P$, such that 
$t\mapsto B_t^{-1}S_t$ becomes a martingale. Under this new probability
measure 
\[
S_t=S_0 \exp\left((r-\frac{\sigma^2}{2} )t + \sigma W^*_t\right)\,,
\]
where $W^*_t=W_t+\frac{r-\mu-\frac{\sigma^2}{2} }{\sigma}$ is a Brownian
motion under $\P^*$.

This new probability measure is now used to price derivatives in this
model: if $C$ is some European contingent claim, that is, a derivative
whose payoff $C_T$ at time $T$ is a function of 
$S_t, 0\le t\le T$, then its arbitrage-free
price at time 0 is given by 
\begin{equation}\label{eq:bs-formula1}
C_0=\E^*(B_T^{-1}C_T)\,,
\end{equation}
where $\E^*$ denotes expectation with respect to $\P^*$. When $C_T$ depends
only on finitely many $S_{t_j},\, j=1,\ldots, m$ then the expectation
in \eqref{eq:bs-formula1} can be written as an $m$-dimensional integral,
which is where QMC enters the game. The details of this will be given 
in Section \ref{sc:generation-brownian}.

In our continuous time model we assume that the option can be traded
at any time prior to its maturity $T$.  
For this, the time $t$ analog of \eqref{eq:bs-formula1} is 
\begin{equation}\label{eq:bs-formula2}
B_t^{-1}C_t=\E^*(B_T^{-1}C_T)\,,
\end{equation}
or $C_t=B_t \E^*(B_T^{-1}C_T)$.

Because of its simplicity, the Black-Scholes model does not provide us with
many interesting examples for simulation. One step towards demanding 
problems is to look at the $m$-dimensional Black-Scholes model.

Consider $m$ shares $S^1,\ldots,S^m$ whose price processes are given by
\[
S^j_t=S^j_0 \exp\left(\mu_j t +\sum_{l=1}^k \sigma_{jl}W^l_t\right)
\]
where $W^1,\ldots,W^k$ are $k$ independent Brownian motions and 
$\sigma=(\sigma_{jl})_{jl}$ is a $m\times k$ matrix.
In this model neither the existence nor the uniqueness of a probability measure
that makes each process $(e^{-rt}S^j_t)_{0\le t\le T}$ a martingale is granted.
In fact, every solution $\nu\in\R^m$ of the linear system
\[
\sigma \nu=r \mathbf{1} -\mu -\frac{1}{2}\diag(\sigma\sigma^\top)
\]
gives rise to such a measure ($\mathbf{1}$ is the vector in $\R^m$ with all
entries equal to $1$).

If such a solution exists, the price processes take on the form
\[
S^j_t=S^j_0 \exp\left((r-\frac{1}{2}(\sigma\sigma^\top)_{jj} )t +\sum_{l=1}^k \sigma_{jl}\tilde W^l_t\right)\,,
\]
where 
$\tilde W^1,\ldots,\tilde W^k$ are $k$ independent Brownian motions under 
the new measure. 

\begin{remark}
The new probability measure is equivalent to the original one only if 
we restrict to finite time intervals $[0,T]$, $T<\infty$.
\end{remark}

These models are interesting from the point of view of (optimal) portfolio 
selection, but they also provide us with practical high-dimensional 
integration problems through derivative pricing. Important examples are 
{\em basket options}\index{basket options}, which are derivatives whose
payoff depends on the price process of several shares. One example of a payoff 
of a basket option on shares with prices $S^1,\ldots,S^d$ is 
\[
C_T=\max\big(w^1S^1_T+\ldots+w^rd1S^d_T-K,0\big)\,,
\]
for some weights $w^1,\ldots,w^d$.

\subsection{SDE models}

In many models from financial mathematics, the share price process
is not given explicitly but is described via a {\em stochastic differential equation}\index{stochastic differential equation}, in short {\em SDE}\index{SDE}.

For example, the SDE corresponding to the basic Black-Scholes model is
\[\begin{aligned}
dS_t&=\hat\mu S_t dt +\sigma S_t dW_t\\
S_0&=s_0\,.
\end{aligned}
\]
The a.s. unique solution\footnote{This is a consequence of the famous It\^o 
formula from stochastic analysis. In short, the It\^o formula states that for
a function $f$ which is $C^1$ in the first variable and $C^2$ in the second variable, we have 
\[
d f(t,W_t)=\frac{\partial f}{\partial t}(t,W_t)dt 
+\frac{\partial f}{\partial W}(t,W_t)dW_t
+\frac{1}{2}\frac{\partial^2 f}{\partial W^2}(t,W_t)dt\,.
\]} to this SDE with initial value $S_0$ is 
\[
S_t=s_0 \exp(\hat\mu t +\sigma W_t -\frac{\sigma^2}{2}t)\,,
\]
such that for $\hat \mu=\mu +\frac{\sigma^2}{2}$ we recover the
price process from \eqref{eq:bs-model}.

More generally, a model could be defined by an $m+1$-dimensional 
SDE 
\begin{equation}\label{eq:gen-sde}
\begin{aligned}
dS_t&=\mu(t,S_t)dt+\sigma(t,S_t)dW_t\\
S_0&=s_0\,.
\end{aligned}
\end{equation}
where $S=(S^0,\ldots,S^m)$ is an $m+1$-dimensional stochastic process 
and $s=(s^0,\ldots,s^m)\in \R^{m+1}$. It is assumed that one coordinate is 
the price of
an asset that can function as a numeraire in that it is never 0. 
In this general model not all the components need
to correspond to share prices or indeed to prices at all. 
Consider, for example the so-called Heston model (already under an equivalent 
martingale measure):
\[
\begin{aligned}
dB_t&=r B_t dt\\
dS_t&=r S_t dt+\sqrt{V_t} S_t (\rho dW^1_t +\sqrt{1-\rho^2}dW^2_t) \\
dV_t&=\kappa (\theta-V_t) dt+\xi \sqrt{V_t}\, dW^1_t \\
(B_0,S_0,V_0)&=(b_0,s_0,v_0)\,.
\end{aligned}
\]  
Here, $r, \kappa,\theta,\xi$ are positive constants, $\mu$ is a real constant,
and $-1<\rho<1$ is a correlation coefficient. 

The third component of our process, $V$, is the so-called volatility of the
share price and is not a tradable asset.

It is worth mentioning that, despite there not being an explicit solution 
known for the SDE, there is a semi-exact formula for the price of 
a European call option in the Heston model using Laplace inversion. \\ 
 
We do not concern ourselves with the theory of SDEs since this is clearly 
beyond the scope of our article. 
From the point of view
of (quasi-)Monte Carlo it is mostly of interest to know that under suitable regularity requirements on the
coefficients of the SDE there exists a unique solution and that under even
stronger conditions this solution can be approximated. 

Let $S_T$ be the solution to the SDE at time $T$ and let $\hat{S}_N$ be
some approximation to $S_T$ computed on the time grid $0=t_0<t_1<\ldots<t_N=T$
with fineness $\delta=\max_{1\le k\le N}(t_k-t_{k-1})$.
We say that $\hat S_N$ converges to $S_T$ in the {\em strong sense} with 
order $\gamma$, if $\E(|S_T-\hat S_N|)=O(\delta^\gamma)$.

Sometimes it is enough to 
compute some characteristics
of the solution like $\E(f(S_T))$ for a function $f$ belonging to some class
$C$. This question is linked to the concept of {\em weak convergence} of 
numerical schemes. See, for example, \cite[Chapter 9.7]{kloeden}. The benefit
is that the weak order of an approximation scheme is usually higher than the
strong order of the same scheme.

The most straightforward solution method is the Euler-Maruyama method:
given \eqref{eq:gen-sde} we compute an approximate solution $\hat S$ on the 
time nodes
$0,h,\ldots, n h=T$ via
\begin{align}
\nonumber
\widehat S_0&= S_0\\
\label{eq:euler}
\widehat S_{k+1}&=\widehat S_k+\mu(kh,\widehat S_k)h+\sigma(kh,\widehat S_k)\Delta W_{k+1}.
\end{align}
It follows from the definition of Brownian motion that $W_{(k+1)h}-W_{kh}$ is a
normal random vector with expectation 0 and covariance matrix
$\sqrt{h}\mathbf{1}_{\R^{m+1}}$. Frequently, \eqref{eq:euler}
is therefore stated in the form
\begin{equation}
\hat S_{k+1}=\hat S_k+\mu(kh,\hat S_k)h+\sigma(kh,\hat S_k)\sqrt{h}Z_{k+1}\,,
\end{equation}
where $Z_1,Z_2,\ldots$ is a sequence of standard normal vectors.
However, we will prefer the original form when using quasi-Monte Carlo.

Under suitable regularity conditions (Lipschitz in second variable, 
sublinear growth with first variable, sufficient smoothness) 
on the coefficient functions
$\mu,\sigma$ of the SDE,  the Euler Maruyama scheme converges in the strong sense with order $\frac{1}{2}$ and in the weak sense  with order $1$, such that,
for sufficiently regular $f$, $\E(f(\hat S_{nh}))$ is a decent approximation to 
$\E(f(S_T))$, for sufficiently small $h$. Discussion of the regularity 
conditions and proofs can be found in \cite{kloeden}.

We report two other schemes for solving autonomous 
SDEs numerically, which  under 
appropriate conditions on the coefficients converge in the strong sense with 
order 1. The first is the Milstein scheme,
\begin{equation}\label{eq:milstein}
\hat S_{k+1}=\hat S_k+\mu(\hat S_k)h+\sigma(\hat S_k)\Delta W_{k+1}
+\frac{1}{2}\sigma(\hat S_k)\sigma'(\hat S_k)(\Delta W_{k+1}^2-h)\,,
\end{equation}
where $\Delta W_{k+1}:=W_{(k+1)h}-W_{kh}$ and where $\sigma'$ is the derivative
of $\sigma$. 
The second is an example for a Runge-Kutta scheme, with the advantage of
not requiring a derivative:
\begin{equation}\label{eq:rungekutta}
\hat S_{k+1}=\hat S_k+\mu(\hat S_k)h+\sigma(\hat S_k)\Delta W_{k+1}
+\frac{1}{2}(\sigma(Y_k)-\sigma(\hat S_k))(\Delta W_{k+1}^2-h)h^{-\frac{1}{2}}\,,
\end{equation}
where the supporting value $Y_k$ is given by 
$Y_k=\hat S_k+\sigma(\hat S_k)h^{\frac{1}{2}}$.

A problem that can occur in practice is that the simulated path can leave the 
domain of definition while the exact solution does not. For example, 
the approximate stock price and/or the volatility process may become
negative. See again \cite{kloeden} and also \cite{andersen} for a thorough
treatment of Monte Carlo simulation of the Heston model.

\subsection{L\'evy models}

L\'evy processes are generalizations of Brownian motion.  
The mathematical definition will be given in Section \ref{sec:levy}.

These processes are interesting 
for financial modeling since they allow for jumps. 
In analogy to the Gaussian models, i.e. models built on Brownian motion,
they come in two flavors. There are explicit models where the stock price is 
exponential L\'evy motion:
\[
S_t=\exp(L_t)\,,
\]
where $L$ is a L\'evy process with $\E(\exp(L_t))<\infty$r. Alternatively,
the stock price might 
again be given by an SDE, i.e.
\[
dS_t=f(t,S_{t-})dL_t\,.
\]
If it is possible to sample from the increments of $L$, then the Euler-Maruyama
scheme still allows us to simulate a discrete approximation to the solution $S$,
\[
\hat S_{k+1}=\hat S_k+f(kh,\hat S_k)(L_{(k+1)h}-L_{kh})\,.
\] 
From the point of view of option pricing it is important that
the market is arbitrage-free. That is, we need to find an
equivalent probability measure, such that discounted prices of
tradable assets are martingales. This is usually achieved with the
so-called Esscher transform, a change of measure under which the 
Process $L$ is again a L\'evy process, see for example \cite[Chaper 9.5]{cont}.

\subsection{Examples}

We conclude this very short introduction to financial mathematics  with
some examples.  

A European Call option on a share with price process $(S_t)_{t\ge 0}$ 
and with strike $K$ and maturity $T$ has payoff 
$C_T=\max(S_T-K,0)$. 
The pricing equation \eqref{eq:bs-formula1} therefore gives the
option price in the Black-Scholes model at time $t$ as
\[
C_0=e^{-rT}\E^*( \max(S_T-K,0))\,.
\]
Since
\[ 
S_T=S_0 \exp\left((r-\frac{\sigma^2}{2} )T + \sigma W^*_T\right)\,,
\]
and since $(r-\frac{\sigma^2}{2} )T+\sigma W^*_T$ is a $N(0,\sigma^2 T)$ random variable, we get
\begin{align*}
C_0&=e^{-rT}\int_{-\infty}^\infty \max(S_0e^x-K,0))e^{-\frac{(x-(r-\frac{\sigma^2}{2} )T)^2}{2\sigma^2 T}}\frac{1}{\sqrt{2\pi \sigma^2 T}}dx\\
&=e^{-rT}\int_{\log(\frac{K}{S_0})}^\infty (S_0e^x-K)e^{-\frac{(x-(r-\frac{\sigma^2}{2} )T)^2}{2\sigma^2 T}}\frac{1}{\sqrt{2\pi \sigma^2 T}}dx\,.%\\
%&=e^{-rT}S_0\int_{\log(\frac{K}{S_0})}^\infty e^xe^{-\frac{(x-(r-\frac{\sigma^2}{2} )T)^2}{2\sigma^2 T}}\frac{1}{\sqrt{2\pi \sigma^2 T}}dx
%\\&\;\;\;+e^{-rT}K\int_{\log(\frac{K}{S_0})}^\infty e^{-\frac{(x-(r-\frac{\sigma^2}{2} )T)^2}{2\sigma^2 T}}\frac{1}{\sqrt{2\pi \sigma^2 T}}dx\\
\end{align*}
The integral can in fact be computed and its value is given by the famous
Black-Scholes option pricing formula
\begin{equation}\label{eq:the-bs-formula1}
C_0=S_0\Phi(d_1)-e^{-rT}K\Phi(d_2)\,,
\end{equation}
where $\Phi(x)=\int_{-\infty}^x e^{-\frac{x^2}{2}}\frac{1}{\sqrt{2\pi}}dx$ and
\begin{equation}\label{eq:the-bs-formula2}
d_1=\frac{\log{\frac{S_0}{K}}+(r+\frac{\sigma^2}{2})T}{\sigma \sqrt{T}}
\quad\text{and}\quad d_2=\frac{\log{\frac{S_0}{K}}+(r-\frac{\sigma^2}{2})T}{\sigma \sqrt{T}}\,.
\end{equation}
So in this case we get a closed-form formula and there is no need to apply
simulation techniques. The price $C_t$ for $0\le t\le T$ can be obtained
from equations \eqref{eq:the-bs-formula1} and \eqref{eq:the-bs-formula2}
simply by substituting $(T-t)$ for $T$.  
Another class of examples for which there often exist closed formulas are
barrier- and lookback options, where the payoff depends on the maximum
or minimum of the price over a given interval.

We move on to a somewhat harder example: the payoff of an {\em Asian option} 
\index{Asian option} written on a share with price process $(S_t)_{t\in [0,T]}$ depends on
the average price over some interval $[T_0,T]$, $T_0< T$, where $T$ 
is the expiry date of the option.  The payoff of a {\em fixed strike} 
Asian call option is given by 
\[
C^{\mathrm{fix}}_T=\max \left(\frac{1}{T-T_0}\int_{T_0}^{T} S_\tau d\tau-K,0\right)\,, 
\]
The payoff of a {\em floating strike} 
Asian call option is given by 
\[
C^\mathrm{flt}_T=\max \left(\frac{1}{T-T_0}\int_{T_0}^{T} S_\tau d\tau-S_T,0\right)\,.
\]
Up to now, nobody has found an explicit formula for either Asian option, but 
there are rather efficient methods using PDEs to compute the value, see 
for example \cite{rogers}. Nevertheless, this example is a nice benchmark for 
simulation methods.\\  

For basket options on several shares the PDE method becomes intractable. Here,
we really have to use simulation. A possible example
payoff is 
\[
\max\left(\frac{1}{m}(S^{1}_T+\ldots+S^{m}_T)-K,0\right) \,,
\] but
more complicated dependencies on the price processes can be encountered in 
practice. In particular, the payoff may depend on the time-averages of the price
processes. Then the option  also has some Asian characteristics.

\section[MC and QMC simulation]{Monte Carlo and quasi-Monte Carlo simulation}
\label{sec:simulation}

\subsection{Non-uniform random number generation}

Most random variables encountered in practical models are not uniformly 
distributed.
We are therefore interested in methods for generating pseudo- or quasi-random 
numbers with a given distribution from their uniform counterparts.

The most straightforward method is the so-called inversion method which 
will be presented in the first subsection.

We are also going to present the class of acceptance-rejection methods for
generating random numbers with a given distribution. We will also argue that
these methods, while usually being the most efficient for Monte Carlo, are not
suited for quasi-Monte Carlo.

\subsubsection{Inversion method}

The most straightforward method for constructing non-uniform
pseudo random numbers from uniform ones is the inversion method.

We introduce this method for a special case only\exclude{and defer the general method to the exercises}. Consider a real random variable $X$ 
with bijective cumulative distribution
function (CDF) $F$, i.e.
$F:\R\longrightarrow (0,1)$,
$F(x)=\P(X\le x)$ for all $x\in \R$ is such that there exists 
$G:(0,1)\longrightarrow\R$ with $G(F(x))=x$ for all $x\in \R$ and 
$F(G(u))=u$ for all $u\in (0,1)$.

Suppose now that the random variable $U$ is uniformly distributed on $(0,1)$
and define a real random variable $Y:=G(U)$. 
Then $Y$ has the same distribution as $X$. To see this, let
$y\in\R$. Then 
\begin{align*}
\P(Y\le y)&=\P(G(U)\le y)=\P\big(F(G(U))\le F(y)\big)
= \P(U\le F(y))=F(y)\,.
\end{align*}
So $F$ is also the distribution function of $Y$.

A sufficient condition for a cumulative distribution function 
to be invertible is that it has a positive probability density function (PDF)
on $\R$. 
\exclude{See Exercise \ref{ue:cdfinv}.}

\subsubsection{Acceptance-rejection method}

Inverting a CDF numerically can be computationally expensive. 
A very versatile and cheap alternative method for generating a random variable with
prescribed probability density function $f$ is the acceptance-rejection method.
For its implementation we need another distribution for which it is cheap
to sample from, e.g., via the inversion  method. Let $g$ be the probability 
density function of this distribution. Moreover, we need that, for some 
$c>0$, $f(x)\le c g(x)$ for all $x\in \R$.

\exclude{Let $Y_1,Y_2,\ldots$ be a sequence of independent 
random variables distributed according
to density $g$ and let $U_1,U_2,\ldots$ be a sequence of independent random 
variables which are uniformly distributed on $[0,1]$.

Let $N:=\min\left\{k:U_k\le  \frac{f(Y_k)}{cg(Y_k)}\right\}$ and let $X=Y_N$. 
\begin{align*}
\P\left(U_k\le  \frac{f(Y_k)}{cg(Y_k)}\right)
&=\int_{-\infty}^\infty \P\left(U_k\le  \frac{f(Y_k)}{cg(Y_k)}\big|Y_k=y\right)g(y)dy\\
&=\int_{-\infty}^\infty  \frac{f(y)}{cg(y)}g(y)dy= \frac{1}{c}\,.
\end{align*}
Therefore, $N$ has geometric distribution with parameter $1-\frac{1}{c}$,
$\P(N=k)=\frac{1}{c}(1-\frac{1}{c})^{k-1}$.
\begin{align*}
\P\left(Y_k\le x \text{ and } U_k\le  \frac{f(Y_k)}{cg(Y_k)}\right)
&=\int_{-\infty}^x \P\left(U_k\le  \frac{f(Y_k)}{cg(Y_k)}|Y_k=y\right)g(y)dy\\
&=\int_{-\infty}^x \frac{1}{c}f(y)dy= \frac{1}{c} \int_{-\infty}^x f(y)dy\,.
\end{align*}
Now
\begin{align*}
\P(X\le x)
&= \sum_{k=1}^\infty \P(X\le x | N=k)\P(N=k)\\
&= \sum_{k=1}^\infty \P(Y_k\le x | N=k)\P(N=k)\\
&= \int_{-\infty}^x f(y)dy \sum_{k=1}^\infty  (1-\frac{1}{c})^{k-1} \frac{1}{c}\\
&= \int_{-\infty}^x f(y)dy \,.
\end{align*}
}

The algorithm is as follows:
\begin{algorithm}
\begin{enumerate}
\item Generate a sample $Y$ from density $g$ and a uniform random variable $U$. 
\item If $U\le  \frac{f(Y)}{cg(Y)}$, set $X=Y$ else go back to step 1.
\end{enumerate}
\end{algorithm}

It is not hard to give a proof that the algorithm gives indeed a random
variable with the desired distribution,
and it follows from the proof that $c$ should be as small as possible
so that the algorithm stops after only few steps.

\subsubsection{Box-Muller method and Marsaglia-Bray algorithm}

Recall the definition of a normal (or Gaussian) random variable:
\index{Gaussian random variable}\index{normal random variable}\index{random variable!Gaussian}
\index{random variable!normal}

\begin{definition}
A random variable $X$ is {\em normally distributed} with mean $\mu$ and
variance $\sigma^2>0$ if it has probability density function
\[
f_X(x)=
\frac{1}{\sqrt{2 \pi \sigma^2}}
\exp\left(-\frac{(x-\mu)^2}{2 \sigma^2}\right)\,.
\]
More generally, a random vector $X=(X_1,\ldots,X_d)$ is said to be {\em normally
distributed} with mean $\mu\in \R^d$ and covariance matrix $\Sigma>0$ if
it has joint probability density function
\[
f_X(x)=
\frac{1 }{\sqrt{(2 \pi) ^d\det(\Sigma)}}
\exp\left(-\frac{(x-\mu)^\top \Sigma^{-1}(x-\mu)}{2}\right)\,.
\]
Here, $\Sigma>0$ means that $\Sigma$ has to be positive definite,
i.e. $x^\top \Sigma x>0$ for all $x\in \R^d\backslash\{0\}$.
\end{definition}

Consider a 2-dimensional standard normal vector $(X,Y)$. 
\begin{align*}
\P(\sqrt{X^2+Y^2}\le r)&=\int_{-r}^r\int_{-\sqrt{r^2-y^2}}^{\sqrt{r^2-y^2}}
\exp(-(x^2+y^2)/2) /\sqrt{2\pi} dx\,dy\\
&=\int_0^r \int_{0}^{2\pi} \rho \exp(-\rho^2/2) /(2\pi) d\varphi\,d\rho\\
&=1-\exp(-r^2/2)\,.
\end{align*}

It follows that the modulus of $(X,Y)$ has distribution function 
$F_R(r)=1-\exp(-r^2/2)$. But that means that we can generate
a random radius by inversion of $F_R$, 
$F_R^{-1}(u)=\sqrt{-2\log(1-u)}$

\begin{algorithm} \label{alg:boxmuller}
\begin{enumerate}
\item Generate two independent $U[0,1)$ random samples $U,V$;
\item let $R=\sqrt{-2\log(1-U)}$;
\item let $X=R \cos(2 \pi V)$ and $Y=R \sin(2\pi V)$.
\end{enumerate}
\end{algorithm}

\begin{remark}
In Algorithm \ref{alg:boxmuller} we could have let 
$R=\sqrt{-2\log(U)}$ as well. But many implementations of pseudo-random 
number generators
give, with very low but still positive probability 0 while never giving
1. So having $1-U$ as the argument of
the logarithm is slightly saver.
\exclude{If we are going to use quasi-Monte Carlo, then our caution is even more
justified since, for example, digital nets give zero sometimes. Also,
we cannot easily discard a point with a zero component since this 
destroys the structure of the point set.}
\end{remark}

There is a acceptance-rejection-type variant of the Box-Muller method which is 
known as Marsaglia-Bray algorithm:

\begin{algorithm}[Marsaglia-Bray]
\begin{enumerate}
\item Generate two independent $U[0,1)$ random samples $U,V$;
\item let $U_1=2 U-1$ and $V_1=2 V-1$;
\item if $U_1^2+V_1^2\ge 1$ reject $(U,V)$ and start from the beginning;
\item else let $S=U_1^2+V_1^2$; 
\item if $S=0$ set $(X,Y)=(0,0)$; 
\item else set $X=U_1 \sqrt{-2\log(S)/S}$ and $Y=V_1\sqrt{-2\log(S)/S}$\,.
\end{enumerate}
\end{algorithm}

We leave the proof that $(X,Y)$ are independent standard normal variables 
to the reader.

\subsubsection{Importance sampling}

For some densities it is very hard -- if not impossible -- to invert the 
CDF exactly, and frequently it is very expensive to do so numerically.

On the other hand, it is not always necessary to generate exactly from the
given distribution but rather one samples from a distribution that is close
(in some sense that remains to be made precise) to it and adjusts for the error
made. This method is called \notion{importance sampling} or, in the present 
context,
\notion{smooth rejection}.

We present the idea in a one-dimensional setup, the general case is 
straightforward. Consider a random variable $X$ with PDF $f_X$ and
suppose we want to compute $\E(h(X))$ for some function $h$.
Let $F_X$ denote the corresponding CDF, $F_X(x)=\int_{-\infty}^x f_X(\xi)d\xi$.
Normally, we would compute
\[
\E(h(X))\approx\frac{1}{N}\sum_{n=1}^N h(F_X^{-1}(U_n))
\]
using the inversion method, where $U_1,\ldots, U_N$ is a uniform 
pseudo-random sequence or a low-discrepancy sequence.

Suppose now that we do not know how to (cheaply) invert $F_X$. 

In addition, assume that there is another PDF $g$ for which $G$, 
$G(x)=\int_{-\infty}^x g(\xi)d\xi$ is easily inverted. Then
\begin{align*}
\E(h(X))
&=\int_{-\infty}^\infty h(x)f_X(x)dx\\
&=\int_{-\infty}^\infty h(x)\frac{f_X(x)}{g(x)}g(x)dx\\
&=\E\left(h(Y)\frac{f_X(Y)}{g(Y)}\right)\,,
\end{align*}
where $Y$ is a random variable with PDF $g$.
Now the last expected value can be computed by sampling from the density $h$
using the inversion method.
\[
\E\left(h(Y)\frac{f_X(Y)}{g(Y)}\right)
\approx\frac{1}{N}\sum_{n=1}^N g\left(H^{-1}(U_n)\right)\frac{f_X(H^{-1}(U_n))}{g(H^{-1}(U_n))}
\]

\begin{remark}
When using Monte Carlo, one may also sample from the density $h$ using the
rejection method. The goal of importance sampling is then to 
reduce the variance of the integrand to speed up convergence. See
for example \cite{glasserman}.
\end{remark}

\begin{remark}
Importance sampling is particularly useful for sampling from a random
vector whose components have a complicated correlation structure.
\end{remark}

\subsubsection{Why not to use rejection with quasi-Monte Carlo}

We already mentioned that using rejection algorithms with quasi-Monte Carlo
is not appropriate. This does not necessarily mean that doing so will lead
to wrong results. But the results will be more costly than with Monte Carlo
and less accurate than with QMC without rejection. 

But first consider Monte Carlo simulation. Usually we are given 
a pseudo random number generator that gives us a sequence $(U_n)_{n\ge 1}$
of numbers in $[0,1]$ which are -- ideally -- indistinguishable from 
a truly random sequence of independent random variables with uniform 
distribution on $[0,1)$. From the sequence $(U_n)_{n\ge 1}$ we now
compute a sequence $(X_n)_{n\ge 1}$ of independent random variables with 
given distributions, for example by using a rejection algorithm. 
To the degree that the original sequence obeys the laws of probability
the transformed sequence will do so as well. If on average a fraction 
$\beta$ close to 1 of the original sequence is rejected, that does not hurt much.

For quasi-Monte Carlo the situation is quite different.
If we have a low discrepancy sequence $(u_n)_{n\ge 1}$ in the $s$-dimensional 
unit cube and we apply a rejection algorithm to every component then we 
have to make a decision about what to do if one component is rejected. 
Do we reject the whole point, that is, all the components? 
What else could we do?

No matter what we do, we will loose 
the low-discrepancy structure of the sequence.

We provide a simple example. Let $f$ be the probability density function of
the Gamma distribution with parameter $a$, $f(x)=x^{a - 1} \exp(-x)/\Gamma(a)$  
and let $g$ be the density of the exponential distribution with 
parameter $b$, $g(x)=b \exp(-b x)$. If $a=1.2$ and $b=0.85$, then
$f(x)\le b^{-1}g(x)$. 
We apply the rejection algorithm to some lattice rule in dimension
4, that is, the first two components are used to generate the first 
Gamma-variable while the last two components will be used to generate the 
second one. If rejection occurs in generating either of the components,
the whole 4-dimensional sample is rejected.

The resulting sequence $(x_n)_{n\ge 1}$ will have the distribution of 
two independent
$\Gamma(1.2)$ variables, so applying the corresponding CDF to the 
components gives a sequence $(u_n)_{n\ge 1}$
which is uniform in the unit square. However there is no reason why 
it should have any additional structure, like having low discrepancy or being
a $(t,4)$-sequence. Figure \ref{fig:rejection2} compares  
$(u_n)_{n\ge 1}$ with the first and third component of
the original lattice.
Of course, the whole number of points in the lattice must be 
greater than the number plotted so we can show an equal number of points in both
plots.

It can be seen that, while the points on the left still bear some similarities
to the lattice, but that they show some characteristics typical for random
numbers, i.e., they show the presence of clusters and holes.\\

\begin{figure}[h]
\begin{center}
\includegraphics[scale=0.6]{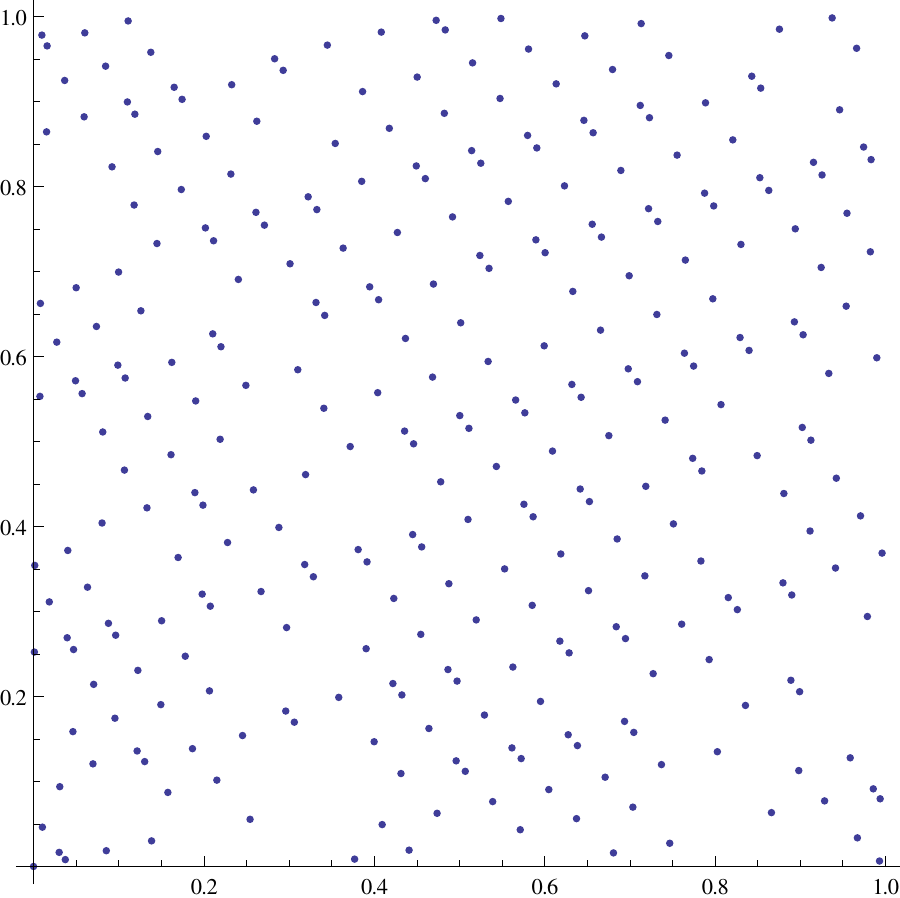}
\includegraphics[scale=0.6]{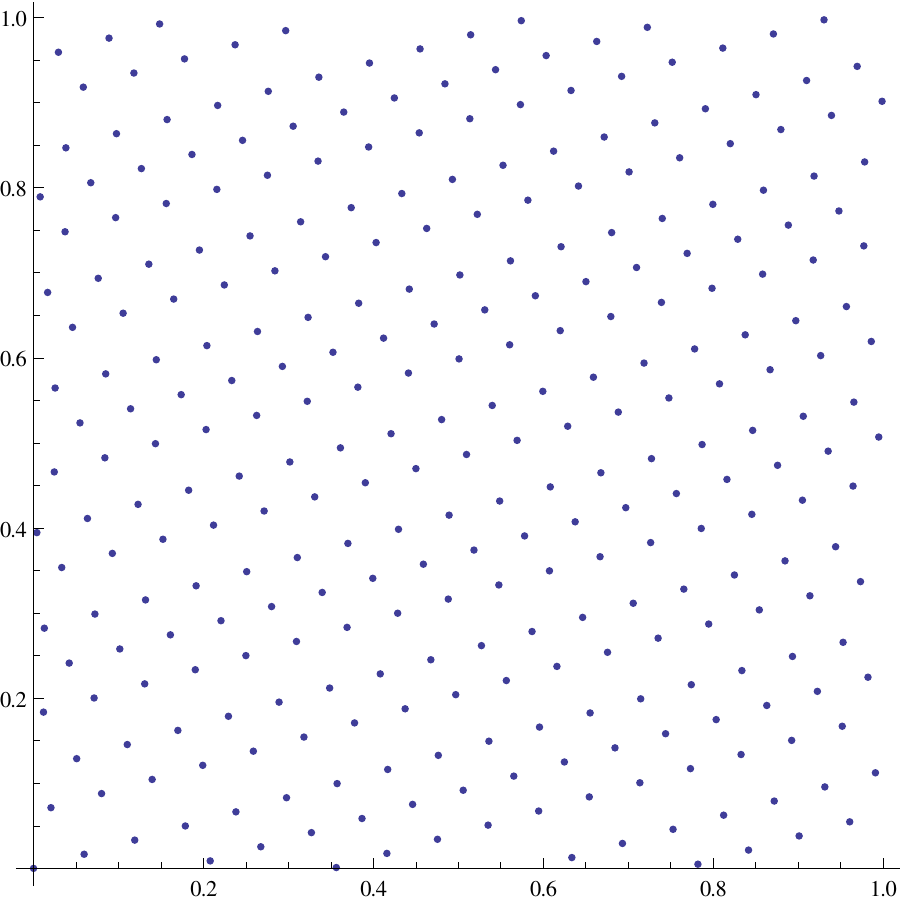}
\end{center}
\caption{Comparison of rejected and original lattice points}\label{fig:rejection2}
\end{figure}

\subsubsection{When still not to use rejection with Monte Carlo}

Another issue with the acceptance-rejection method is that it sometimes makes
the dependence of the result of
a Monte Carlo simulation on the model parameters less smooth. 
It is clear that the result of a true Monte Carlo simulation is
by definition stochastic. If one looks for model parameters 
which minimize (a function of) the integral that is computed,
then this has the paractical drawback that for example Newton's method
cannot be used. In practice it is therefore common to fix the random sequence
for the Monte Carlo simulation, i.e., the random generator is started afresh
for each set of parameters. In this sense the Monte Carlo method becomes
closer to QMC, because the point set is now deterministic.

However, if acceptance-rejection is used for the generation of
random variables, then the integral as a function of the model
parameters can still  be noisy.
The following artificial example is taken from \cite{eileze}.

\begin{example}\label{ex:gamma}
Let $\left(X_i^\lambda\right)_{i=1,\hdots,n}$ be a sequence of i.i.d.
$\mathrm{Gamma}(\lambda,1)$ random variables and $S^\lambda=\sum_{i=1}^n
X^\lambda_i$. Let further $f(s):=s-\bar{\lambda}\cdot n$.
 
We want to approximate
$$
\alpha(\lambda) = \mathbb{E}\left[f(S^\lambda)\right]
$$
by the estimator
$$
\hat\alpha_N(\lambda)=\frac 1N\sum_{j=1}^N f(S^\lambda_j)
$$
for different values of $\lambda$, 
$\lambda\in(\bar{\lambda}-\epsilon,\bar{\lambda}+\epsilon)$.
There are two scenarios:
\begin{enumerate}
\item We use a Monte Carlo method and acceptance-rejection with
a suitable exponential distribution as dominating function. The pseudo 
number generator is restarted for every choice of $\lambda$, so that
in fact we use the same sequence for every integral evaluation.
The reason for this is that otherwise $\hat\alpha(\lambda)$ will
be by itself random.

\item We use a low discrepancy quasi-Monte Carlo sequence (here: a Sobol
sequence) together with the inverse transform method. 
\end{enumerate}

We draw those functions for $n=5,N=1024,\bar{\lambda}=2$ and $\epsilon=0.2$, where $\lambda$ changes in steps of $0.001$.
In Figure \ref{fig:mc_rejection} one can see quite some noise while in
Figure \ref{fig:qmc_invers} the graph is very smooth. 

Smoothness is of importance if, for example, one wants to minimize
$\alpha(\lambda)$. An application would be calibration of a financial 
model to market data.
\end{example}  

\begin{figure}[h!tb]
\centering\includegraphics[width=9cm]{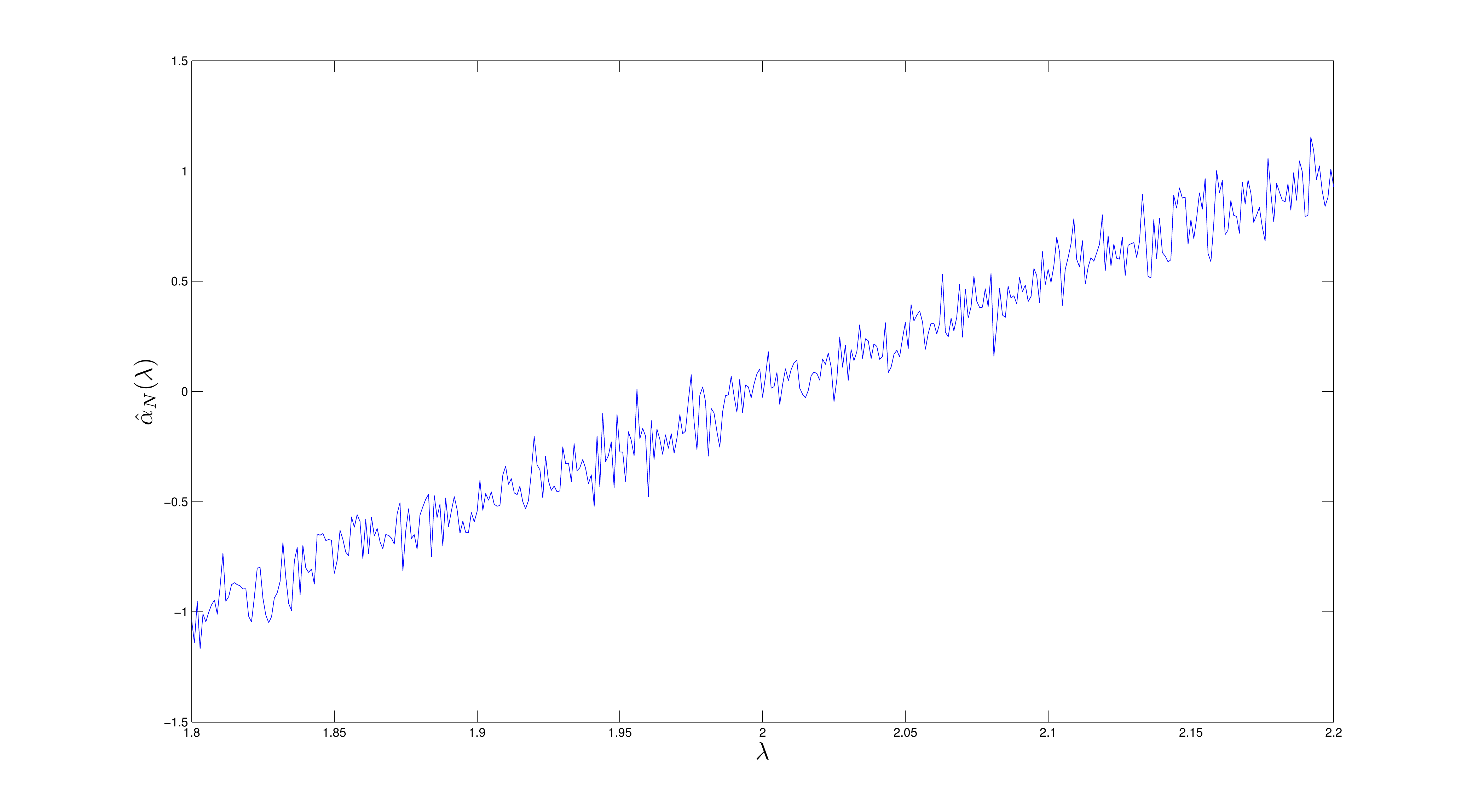}\\
\caption{Acceptance-rejection method with a fixed Monte Carlo point set}\label{fig:mc_rejection}
\end{figure}

\begin{figure}[h!tb]
\centering\includegraphics[width=9cm]{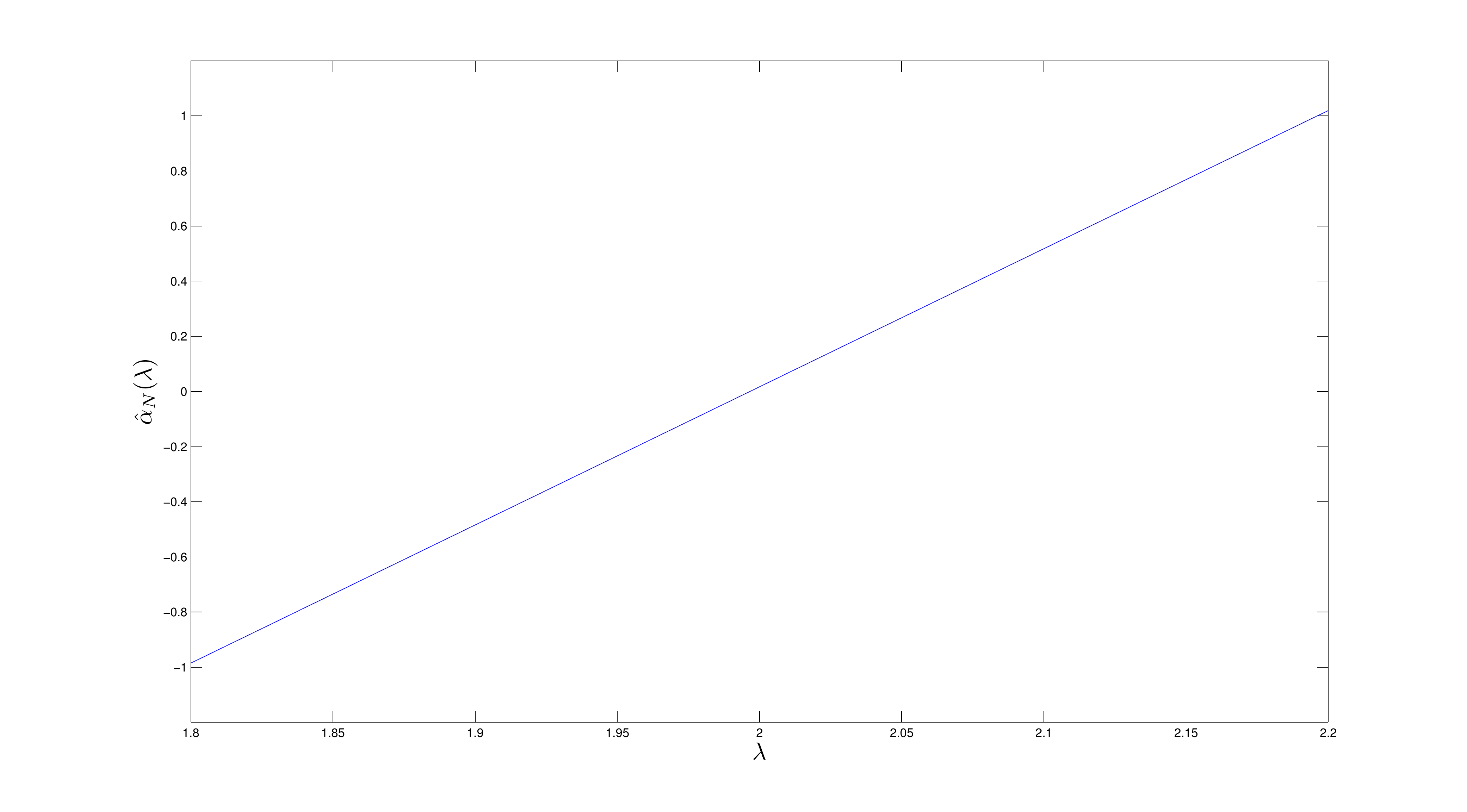}\\
\caption{Inverse transform method with a quasi-Monte Carlo point set}\label{fig:qmc_invers}
\end{figure}

\subsection{Generation of Brownian paths}\label{sc:generation-brownian}

Many problems from finance, but also from physics, encompass phenomena 
which are modeled by a Brownian motion. In this section we give the basic 
definition and describe some methods for sampling from Brownian motion.

\subsubsection{Brownian motion -- definition and properties}
\label{sec:brownian-motion}

\begin{definition}
A {\em standard Brownian motion} $B$ in $\R^d$ is a stochastic process
in continuous time, defined on some  probability space $(\Omega,\Sigma,\P)$, having the following properties:

\begin{enumerate}
\item $B_0=0$ almost surely;
\item $B$ has {\em stationary increments}, that is, for any $s,t\ge 0$ 
the random variables $B_{t+s}-B_t$ and $B_{s}$ have the same distribution; 
\item $B$ has {\em independent increments}, that is, 
for any $n\in\N$ and any $t_1,\ldots,t_n\in [0,\infty)$ with 
$t_0:=0< t_1<t_2<\ldots<t_n$, 
the random variables $B_{t_1}-B_{t_0},\ldots,B_{t_n}-B_{t_{n-1}}$ are
independent;
\item $\sqrt{\frac{1}{t}}B_t$ is a standard normal $\R^d$-valued random variable for every $t\ge 0$; 
\item $B$ has {\em continuous paths}, that is, for each $\omega\in \Omega$ 
the mapping
$t\longmapsto B_t(\omega)$ is continuous.
\end{enumerate}
\end{definition}

For applications we usually only need to evaluate the Brownian path at finitely 
many nodes $t_1,\ldots,t_d$. We therefore define a \notion{discrete
Brownian path} with discretization $0<t_1<\ldots<t_d$ 
as a Gaussian vector $(B_{t_1},\ldots,B_{t_d})$ with mean zero and
covariance matrix 
\[
\big(\min(t_j,t_k)\big)_{j,k=1}^d=
\left(
\begin{array}{cccccccc}
t_1&t_1&t_1&\ldots&t_1\\
t_1&t_2&t_2&\ldots&t_2\\
t_1&t_2&t_3&\ldots&t_3\\
\vdots&\vdots&\vdots&\ddots&\vdots\\
t_1&t_2&t_3&\ldots&t_d
\end{array}
\right)\,.
\]

\subsubsection{Classical constructions}

There are three classical constructions of discrete Brownian paths:
\begin{itemize}
\item the \notion{forward method}\index{forward method}, also known as \notion{step-by-step method} or
\notion{piecewise method} 
\item the \notion{Brownian bridge construction}\index{Brownian bridge construction} or \notion{L\'e{}vy-Ciesielski construction}
\item the \notion{principal component analysis construction}{} (PCA construction)
\index{PCA construction}
\end{itemize}

The forward method is also the most straightforward one: given
a standard normal vector $X=(X_1,\ldots,X_d)$ the discrete Brownian path
is computed inductively by
\[
B_{t_1}=\sqrt{t_1} X_1\,,\quad B_{t_{k+1}}=B_{t_{k}}+\sqrt{t_{k+1}-t_k}X_{k+1}\,.
\]
Using that $\E(X_jX_k)=\delta_{jk}$, it is easy to see that 
$(B_{t_1},\ldots,B_{t_d})$ has the required correlation matrix.
Besides its simplicity, the main attractivity of the forward method lies
in the fact that it is very efficient: given that the values
$\sqrt{t_{k+1}-t_k}$ are pre-computed, generation of a path takes only
generation of the normal
vector plus $d$ multiplications and $d-1$ additions. \\

An alternative construction is the Brownian bridge construction,
which allows the values $B_{t_1},\ldots,B_{t_d}$
to be computed in any given order.
The main observation that makes this possible is the following lemma,
the proof of which is left to the reader.

\begin{lemma}
Let $B$ be a Brownian motion and let $r<s<t$.

Then the conditional distribution of $B_s$ given 
$B_r,B_t$ is $N(\mu,\sigma^2)$ with 
\[
\mu=\frac{t-s}{t-r}B_s+\frac{s-r}{t-r}B_t
\text{ and }\sigma^2=\frac{(t-s)(s-r)}{t-r}\,.
\] 
\end{lemma}

Suppose the elements of $(B_{t_1},\ldots,B_{t_d})$ should be computed
in the order $B_{t_{\pi(1)}},B_{t_{\pi(2)}},\ldots,B_{t_{\pi(d)}}$ for some
permutation $\pi$ of $d$ elements. In computing $B_{t_{\pi(j)}}$
we need to take into account the previously computed elements, 
and at most two of those are of relevance, the one next to $\pi(j)$ on the
left and the one next to $\pi(j)$ on the right: 
define for every $j\in\{1,\ldots,n\}$ two sets,
\begin{align*}
L(j)&:=\{k:k<\pi(j) \;\text{and}\;\pi^{-1}(k)<j\}\\
R(j)&:=\{k:k>\pi(j) \;\text{and}\;\pi^{-1}(k)<j\}\,.
\end{align*}

Thus $L$ contains all the indices $k$ that are smaller than $\pi(j)$
and for which $B_{t_k}$ has already been constructed and $R$ contains all the
indices $k$ that are greater than $\pi(j)$ and for which $B_{t_k}$ has already
been constructed. Now define
\begin{align*}
l(j):=\left\{
\begin{array}{cccc}
0 & \text{if}& L_j=\emptyset\\
\max L_j &  \text{if}& L_j\ne\emptyset
\end{array}
\right.\\
r(j):=\left\{
\begin{array}{cccc}
\infty & \text{if}& R_j=\emptyset\\
\min R_j &  \text{if}& R_j\ne\emptyset
\end{array}
\right.
\end{align*}
and set $B_{t_0}=0$, 
\[
B_{t_{\pi(j)}}:=\left\{
\begin{array}{ccccc}
B_{t_{l(j)}}+\sqrt{t_{\pi(j)}-t_{l(j)}}X_j & \text{if}& r(j)=\infty\\[0.9em]
\begin{array}{r}
\frac{t_{r(j)}-t_{\pi(j)}}{t_{r(j)}-t_{l(j)}}B_{t_{l(j)}}
+\frac{t_{\pi(j)}-t_{l(j)}}{t_{r(j)}-t_{l(j)}}B_{t_{r(j)}}\\
+\sqrt{\frac{(t_{\pi(j)}-t_{l(j)})(t_{r(j)}-t_{\pi(j)})}{t_{r(j)}-t_{l(j)}}}X_j
\end{array}
&\text{if}& r(j)<\infty\,,
\end{array}
\right.
\]
where $X=(X_1,\ldots, X_d)$ is a standard normal random vector.

It is easy to check that the 
vector $(B_{t_1},\ldots,B_{t_d})$ constructed in that way has again
covariance matrix $(\min(t_j,t_k))_{j,k}$. 
The functions $l$ and $r$, as well as the factors of 
$B_{t_{l(j)}}$, $B_{t_{r(j)}}$, $Z_j$, do not depend on the random vector 
$X$ so 
they can be pre-computed. In some special
cases the functions $l$ and $r$ can be computed explicitly, for example if
the $\pi(t_j)$ are the first $n$ elements of the van der Corput sequence or
of the $\{k\alpha\}$-sequence with $\alpha=\frac{1+\sqrt{5}}{2}$, see
\cite{lalesch}. Therefore the Brownian bridge construction is also
very efficient: besides the generation of the vector $X$, computation
of one sample uses at most $2 d$ additions and $3 d$ multiplications. 

Moreover, we see that the forward construction is a special case of the
Brownian bridge construction where $\pi$ is the identical permutation.\\

The PCA construction exploits the fact that the correlation matrix
of $(B_{t_1},\ldots,B_{t_d})$ is positive definite and can therefore 
be written in the form $V D V^{-1}$ for a diagonal matrix $D$ with positive 
entries and an 
orthogonal matrix $V$. $D$ can be written as $D=D^\frac{1}{2}D^\frac{1}{2}$,
where $D^\frac{1}{2}$ is the element-wise positive square root of $D$.
Now the PCA construction from a standard normal random vector $X$ is
given by
\[
(B_{t_1},\ldots,B_{t_d})^\top=V D^\frac{1}{2} X\,.
\]
The disadvantage of the PCA for high-dimensional problems is that 
the matrix-vector multiplication, having computational complexity $O(d^2)$,  
becomes comparatively costly. 
Keiner and Waterhouse \cite{kewa} describe an approximate PCA for which the
cost of 
matrix-vector multiplication is $O(d\log d)$.

\subsubsection{What is wrong about the forward construction?}

We have provided three different constructions of Brownian paths with one
standing apart in that it is clearly the most simple one. So why not use
the forward construction for every application?

The answer is that theory predicts a big integration error for QMC
if dimensions are big and the number of integration nodes is of realistic
order, like a couple of millions only. But one may have the hope that if 
only a limited number input parameters have significant importance for
the result, then QMC might behave very similar as in a low dimensional
integration problem.

Figure \ref{fig:brownian-strat} shows the influence of input 
parameters on the whole discrete path. We compare the forward construction
on the left with the Brownian bridge 
construction on the right. In the two upper plots all but the first input 
variables are held fixed. We see that the influence of the first variable 
on the overall behavior of the path (in an informal sense) is bigger
for the Brownian bridge construction. 

In the two lower plots all but the 7th input variables are held fixed.
We see that in the forward construction only values of the path after the
seventh node are influenced, but the overall influence is only slightly 
smaller than that of the first variable. In contrast, the influence of the
seventh variable in the Brownian bridge construction is restricted to 
the third quarter and is much smaller than that of the first variable. 
 
\begin{figure}[h]
\begin{center}
\includegraphics[scale=0.6]{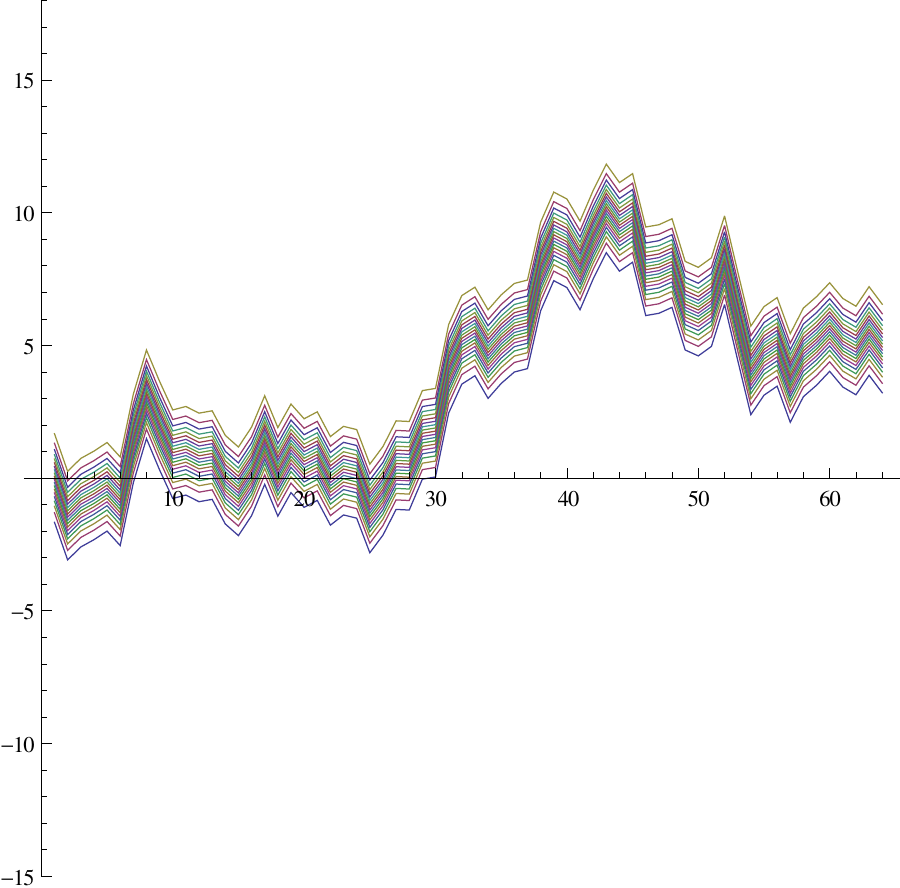}
\includegraphics[scale=0.6]{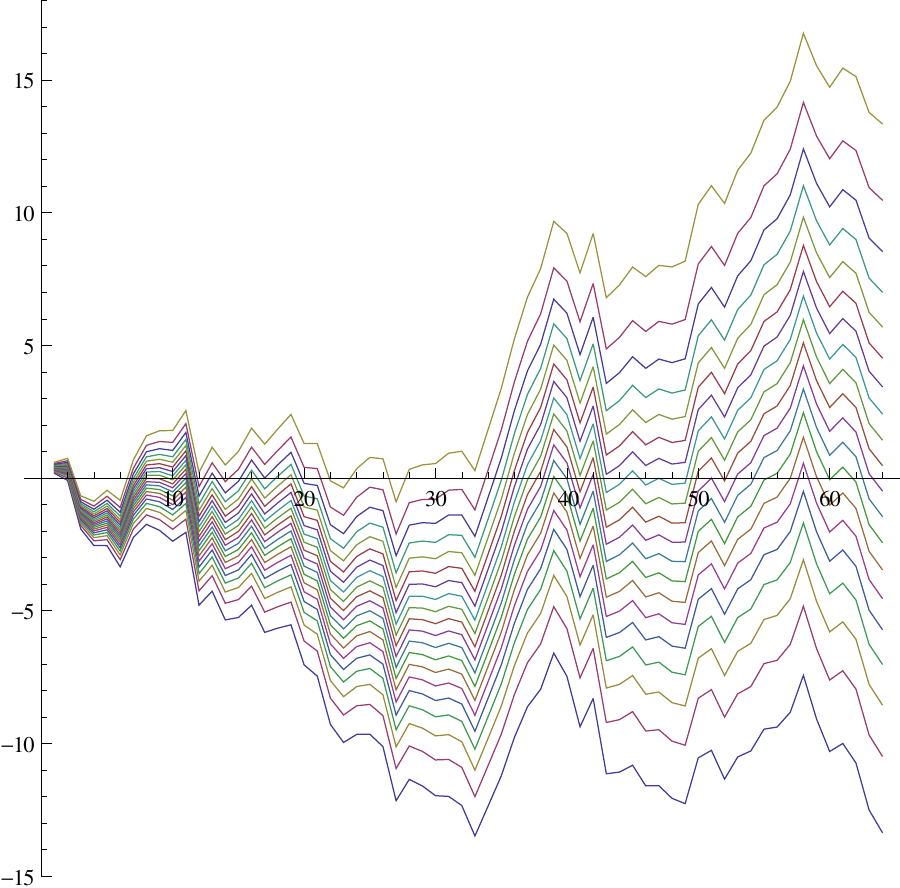}

\includegraphics[scale=0.6]{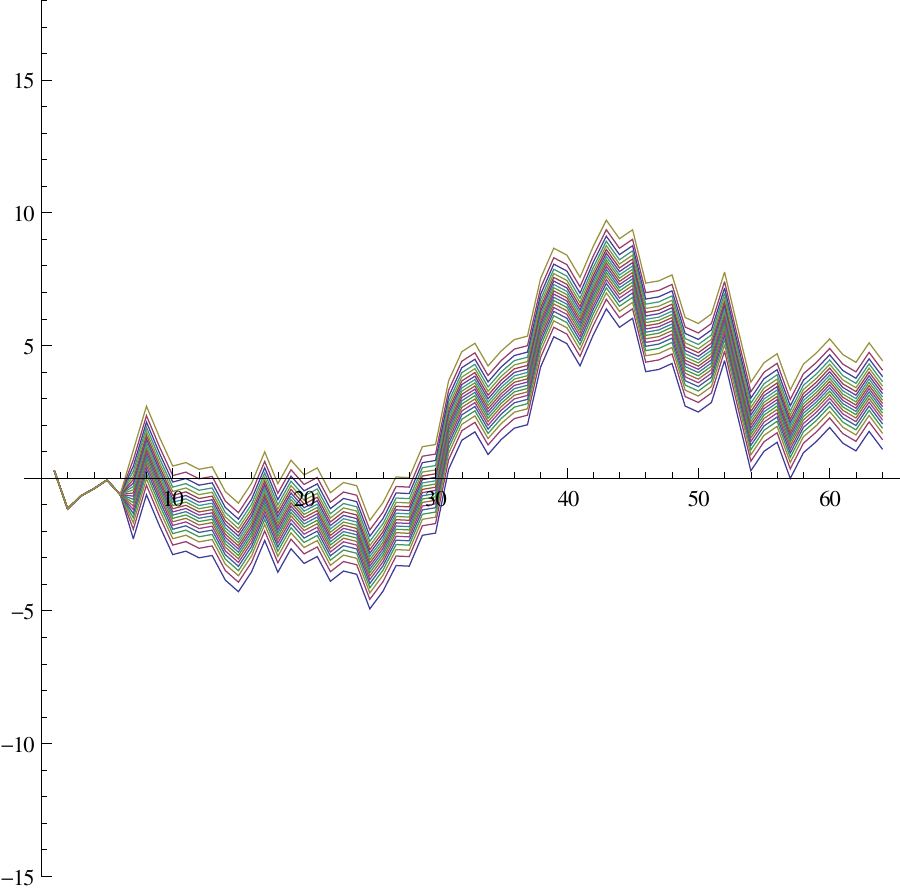}
\includegraphics[scale=0.6]{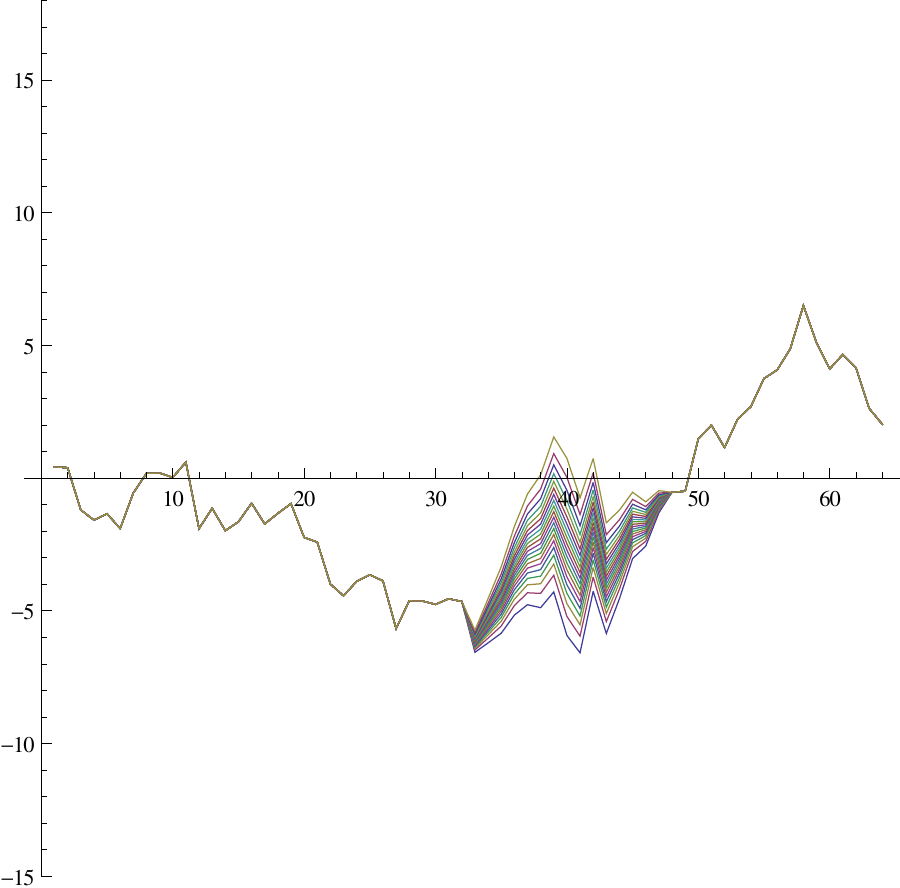}
\end{center}
\caption{Paths of Brownian motion constructed with the forward construction (left) and the Brownian bridge construction (right). All but one parameters are fixed.}\label{fig:brownian-strat}
\end{figure}
The above notion of ``behaving like a low dimensional problem'' is made 
precise in \cite{caflisch97} with the notion of {\em effective dimension}.
It must be added though, that despite of its popularity the concept of
effective dimension alone does not fully explain the success of the alternative
constructions. There is a great number of authors who investigated this problem
and it is still largely unsolved at the present. 

To answer the question posed in the header: there is nothing wrong with the 
forward construction, but for some classes of 
problems other constructions achieve lower
errors, at least empirically. For other problems the forward construction may 
be just fine, as for example in the example due to \cite{papa}, which will
also be one of the examples in Section \ref{sec:examples}.

\subsubsection{Evenly spaced discretization nodes}
The case where the $t_j$ are evenly spaced is 
of special interest as will become apparent soon.
In that case the covariance matrix equals
\[
\Big(\frac{1}{d}\min(j,k)\Big)_{j,k=1}^d=
\frac{1}{d}\left(
\begin{array}{cccccccc}
1&1&1&\ldots&1\\
1&2&2&\ldots&2\\
1&2&3&\ldots&3\\
\vdots&\vdots&\vdots&\ddots&\vdots\\
1&2&3&\ldots&d
\end{array}
\right)\,.
\]
We will denote this matrix by $\Sigma^{(d)}$ or,
if there is no danger of confusion, simply by $\Sigma$.

Note that we can compute the Cholesky decomposition of $\Sigma$ rather
easily:
$\Sigma^{(d)}=S S^\top$, where
\[
S=S^{(d)}:=\frac{1}{\sqrt{d}}
\left(
\begin{array}{cccccccc}
1&0&0&\ldots&0\\
1&1&0&\ldots&0\\
1&1&1&\ldots&0\\
\vdots&\vdots&\vdots&\ddots&\vdots\\
1&1&1&\ldots&1
\end{array}
\right)\,.
\]
Note that if $y=(y_1,\ldots,y_d)$ is a vector in $\R^d$, then 
$Sy$  is the cumulative sum over $y$ divided by $\sqrt{d}$,
\[
Sy=\frac{1}{\sqrt{d}}(y_1,y_1+y_2,\ldots,y_1+\ldots+y_d)\,.
\]
We have the following two easy lemmas:
\begin{lemma}
Let $A$ be any  $d\times d$ matrix with $A A^\top=\Sigma$ and let 
$X$ be a standard normal vector. Then $B=A X$ is a discrete Brownian
path with discretization $\frac{1}{d},\frac{2}{d},\ldots,\frac{d-1}{d},1$.
\end{lemma}

\begin{proof}
%We compute the characteristic function of $AX$:
%\[
%\psi_{AX}(u)=\E\left(e^{iu\cdot AX}\right)
%=\E\left(e^{iA^\top u\cdot X}\right)=\psi_X(A^\top u)=
%=e^{-\frac{(A^\top u)\cdot(A^\top u)}{2}}
%=e^{-\frac{ u\cdot \Sigma u}{2}}
%\]
Since every linear combination of independent normal random variables
is still normal, $AX$ is normal. We compute the covariance matrix:
\begin{align*}
\E\left((AX)_j(AX)_k\right)
&=\E\left(\sum_{l=1}^dA_{jl}X_l\sum_{m=1}^dA_{km}X_m\right)\\
&=\sum_{l=1}^d\sum_{m=1}^dA_{jl}A_{km}\E\left(X_lX_m\right)\\
&=\sum_{l=1}^dA_{jl}A_{kl}=(AA^\top)_{jk}=\Sigma_{jk}\,.
\end{align*}
\end{proof}

\begin{lemma}\label{th:papa}
Let $A$ be any $d\times d$ matrix with $A A^\top=\Sigma$. Then 
there is an orthogonal $d\times d$ matrix $V$ with $A=SV$. Conversely, 
$SV(SV)^\top=\Sigma$ for every orthogonal $d\times d$ matrix $V$.
\end{lemma}

\begin{proof}
Suppose $A A^\top=\Sigma$, such that $AA^\top=SS^\top$. 
Note that $S$ is invertible and  define 
$V=S^{-1}A$. Then 
\[
VV^\top=S^{-1}AA^\top (S^{-1})^\top
=S^{-1}SS^\top (S^{-1})^\top=\mathrm {id}_{\R^d}\,,
\]
showing that $V$ is orthogonal. The converse follows from the
fact that for orthogonal $V$ we have $V^\top=V^{-1}$.
\end{proof}

For evenly spaced discretization nodes the orthogonal matrices 
corresponding to the classical matrices can often be given explicitly.
The orthogonal transform corresponding to the forward method is the identical
mapping on the $\R^d$. For $d=2^k$, the orthogonal transform corresponding
to the Brownian bridge construction where $B$ is computed in the order
$B_1,B_\frac{1}{2},B_\frac{1}{4},B_\frac{3}{4},B_\frac{1}{8},B_\frac{3}{8},B_\frac{5}{8},\ldots$, is given by the inverse Haar transform,
see \cite{leo2012}. For the PCA, the orthogonal transform has been given
by Scheicher, and it has been shown that the computation complexity is
$O(d \log(d))$, see \cite{schei}. The advantage of the representation of $A$
in Lemma \ref{th:papa} is that there are many orthogonal matrices that allow
for fast matrix vector multiplication, that is, a path of length
$d$ can be computed using $O(d\log(d))$ operations. Examples include
the Walsh transform, discrete sine/cosine transform, Hilbert transform and 
others. See again \cite{leo2012}.\\

Coming back to the general case of unevenly spaces discretization nodes
we note the following: suppose you have nodes $0<t_1<\ldots<t_d$. 
We may compute an evenly spaced path $(B_\frac{1}{d},\ldots,B_\frac{1}{d})$ 
using our favorite orthogonal transform, then compute
\[
\tilde B=\sqrt{d}\Big(\sqrt{t_1}B_\frac{1}{d},\sqrt{t_2-t_1}(B_{\frac{2}{d}}-B_{\frac{1}{d}}),\ldots,\sqrt{t_d-t_{d-1}}(B_{\frac{d}{d}}-B_{\frac{d-1}{d}})\Big)\,.
\]
Then $\tilde B$ is a discrete Brownian path with discretization 
$0<t_1<\ldots<t_d$.

\subsection{Generation of L\'evy paths}\label{sec:levy}

\begin{definition}
A {\em L\'evy process} $L$ in $\R^d$ is a stochastic process
in continuous time, defined on some probability space $(\Omega,\cF,\P)$, 
having the following properties:

\begin{enumerate}
\item $L_0=0$ almost surely;
\item $L$ has {\em stationary increments}, that is, for any $s,t\ge 0$,
the random variables $L_{s+t}-L_t$ and $L_s$ have the same distribution; 
\item $L$ has {\em independent increments}, that is, 
for any $n\in\N$ and any $t_1,\ldots,t_n\in [0,\infty)$ with 
$t_0:=0< t_1<t_2<\ldots<t_n$, 
the random variables $L_{t_1}-L_{t_0},\ldots,L_{t_n}-L_{t_{n-1}}$ are
independent;
\item $L$ is {\em continuous in probability}, i.e., for all $t\ge 0$ and
$c>0$,
\[
\lim_{h\rightarrow 0}\P(|L_{t+h}-L_t|>c)=0\,.
\]
\end{enumerate}
Without loss of generality one may also require (see \cite[Chapter I.4, Theorem 30]{protter})
\begin{enumerate}
\setcounter{enumi}{4}
\item $L$ has {\em c\`adl\`ag paths}, that is, for each $\omega\in \Omega$ 
the mapping 
$t\longmapsto L_t(\omega)$ is right-continuous with limits from the left.
\end{enumerate}
\end{definition}

We will concentrate on discrete paths, and therefore properties 4. and 5.
are of minor importance for our purpose. One property that follows from 
the above is that a L\'evy process is already completely characterized by
the distribution of $L_1$. Examples are provided by the Poisson process,
where $L_1$ has Poisson distribution and by Brownian motion, where
$L_1\sim N(0,1)$. 
The L\'evy Kintchine formula (see \cite[Chapter I.4, Theorem 43]{protter})
states that for any L\'evy process there are numbers $b,\sigma\in \R$ and a
measure $\nu$ on $\R\backslash \{0\}$
with $\int_{|x|< 1}x^2 \nu(dx)<\infty$ such that the characteristic function
of $L_t$ is given by
\[
\phi_{L_t}(u)=\E\left(\exp(\icomp uL_t)\right)=\exp(t\,\psi(u))
\]
with
\[
\psi(u)=\icomp bu-\frac{\sigma^2}{2}u^2+\int_{|x|\ge 1}(\exp(\icomp ux)-1)\nu(dx)
+\int_{|x|<1}(\exp(\icomp ux)-1-\icomp ux)\nu(dx)\,.
\]
Thus the distribution of $L_t$ (and therefore of an increment $L_{t+s}-L_t$) 
can be computed via Fourier inversion. For some distributions
like the Normal,
Poisson, and  
Gamma distribution, the density of the increment can be given explicitly.

It is actually straightforward to construct a discrete L\'evy path on a given
set of nodes $0<t_1<\ldots<t_d$ : let $F_t^{-1}$ denote the inverse of the
distribution function of $L_t$. Let $U_1,\ldots,U_d$ be independent 
$U(0,1)$ random variables. Define
\begin{align*}
L_{t_1}&:=F_{t_1}^{-1}(U_1)\\
L_{t_k}&:=L_{t_{k-1}}+F_{t_k-t_{k-1}}^{-1}(U_k)\,.
\end{align*}
That is, the forward method works immediately. The other constructions have no
direct generalizations to L\'evy processes, except for special cases for
which the conditional distribution of $L_m$ given $L_l,L_r$ for $l<m<r$ can
be computed. One such example is the Gamma process, the L\'evy process
for which $L_t$ has gamma distribution with parameters $(t\gamma,\lambda)$,
$\gamma,\lambda\in (0,\infty)$, that is,
\[
\P(L_t\le z)=\int_0^z \frac{x^{\gamma-1}}{\lambda^\gamma \Gamma(\gamma)}\exp(-x/\lambda)dx\,,
\]
where it is shown in \cite{lecuyer} that a Bridge construction is possible
for this process and also for the variance-gamma process, which is a 
L\'evy process of the form $t\mapsto W_{L_t}$, where $L$ is a gamma process
and $W$ is Brownian motion.

However, there is a simple trick, first used in \cite{leo} for the Brownian 
bridge and later, but independently, in \cite{imaitan09} for general orthogonal 
transforms, that recovers some of the qualitative features of those transforms:
we may rewrite the forward construction of the discrete L\'evy Path as
\begin{align*}
L_{t_1}&:=F_{t_1}^{-1}(\Phi(Y_1))\\
L_{t_k}&:=L_{t_{k-1}}+F_{t_k-t_{k-1}}^{-1}(\Phi(Y_k))\,,
\end{align*}
where $Y_1,\ldots,Y_d$ are independent standard normal variables and 
$\Phi$ is the standard normal CDF. The orthogonal transform is now employed
simply in that the $Y_1,\ldots,Y_d$ are generated from
our input variables $X_1,\ldots,X_d$ by multiplication with the orthogonal
matrix, i.e. $Y=VX$. 

Figure \ref{fig:nig-strat} illustrates the effect of this method on the 
construction of discrete normal inverse Gaussian\footnote{Here the increments
have been sampled from the NIG distribution for simplicity. In general, sampling
from $L_t$ for $t\ne 1$ requires Fourier inversion.} (NIG) L\'evy paths. 
The figure on the right shows
the effect of the 7th input variable. In comparison to the corresponding 
Brownian motion example from figure \ref{fig:brownian-strat} we see that
the effect is less localized, but it still the seventh variable mostly
influences the behavior of the path on the interval
$[\frac{1}{2},\frac{3}{4}]$. Note that the plots are slightly misleading
since they interpolate linearly between the discretization points and
thus look like continuous functions. In reality, the paths of an NIG
process are (with probability 1) 
discontinuous with infinitely many jumps in every non-empty open
interval. It is important to keep this in mind if, for example, some 
characteristic of the first entry time of the path into some set is to be 
computed, as is the case, e.g., for barrier options. 

\begin{figure}[h]
\begin{center}
\includegraphics[scale=0.6]{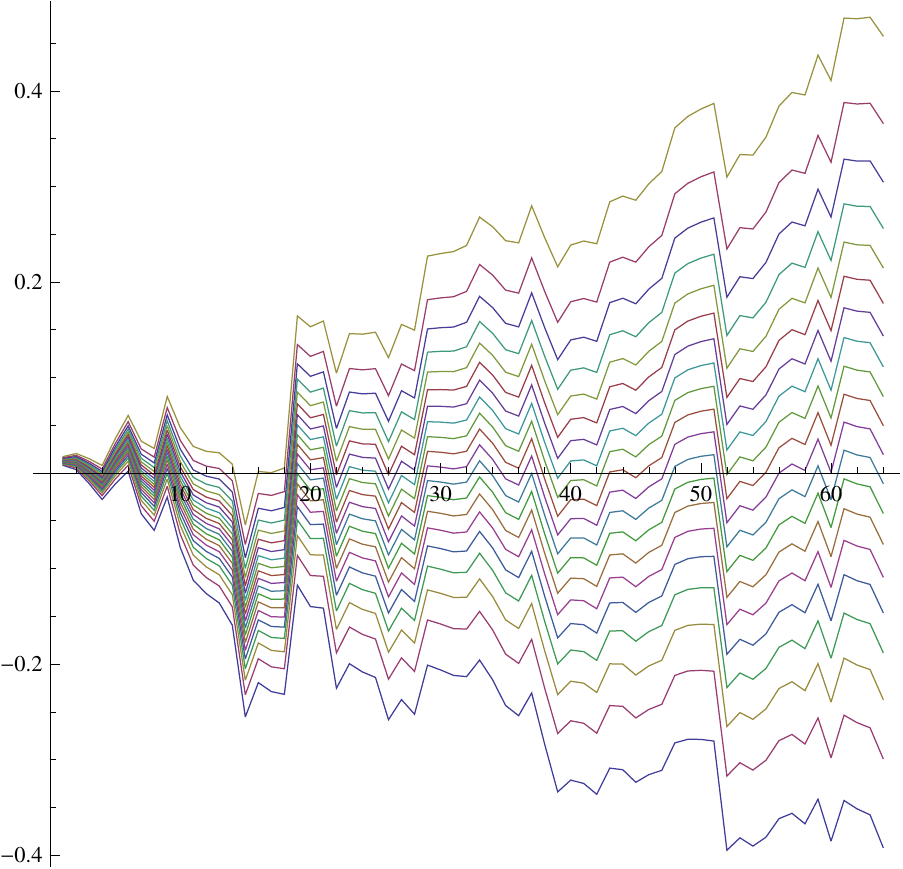}
\includegraphics[scale=0.6]{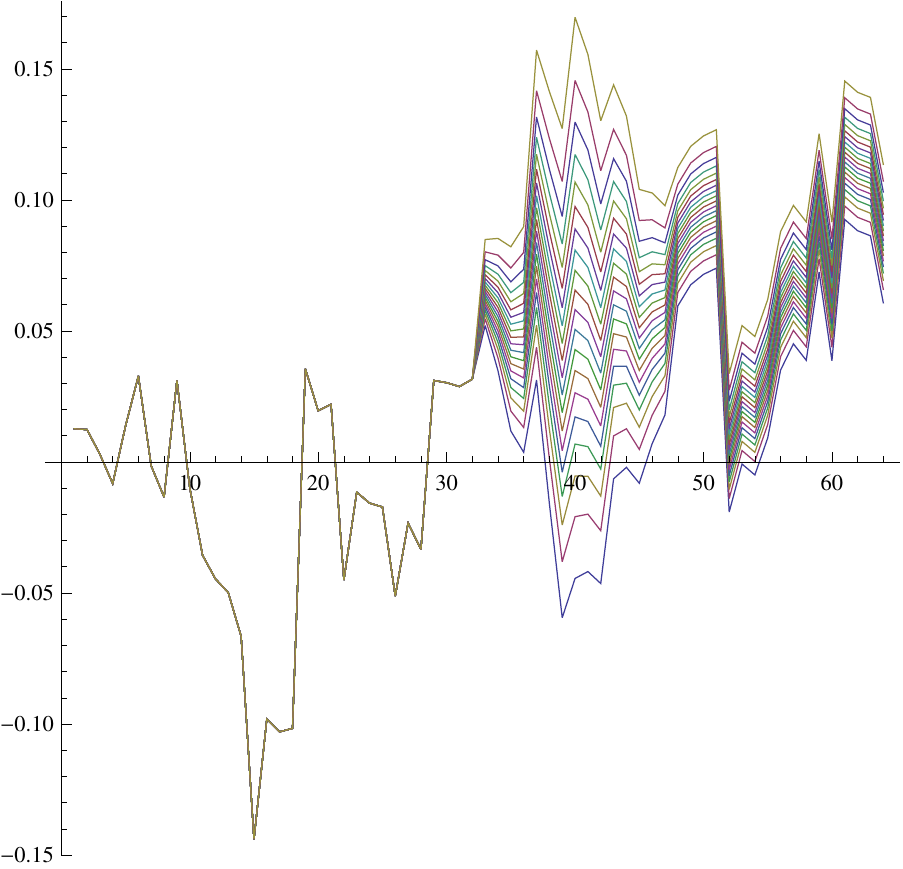}
\end{center}
\caption{NIG process paths constructed with Brownian bridge orthogonal transform. Left figure: all but the first variables held fixed. Right figure: all but the 7th variable held fixed.}\label{fig:nig-strat}
\end{figure}

\subsection{Multilevel (quasi-)Monte Carlo}

Multilevel Monte Carlo is a technique for speeding up Monte Carlo simulation,
especially for SDE models.
It has gained
a lot of recognition over the last couple of years, starting with the pioneering 
work by Giles \cite{giles} and 
Heinrich \cite{heinrich}.
We give a short account of the method.

Suppose we want to approximate $\E(Y)$ for some random variable
$Y$ which has finite expectation. Suppose further that we have a sequence of 
sufficiently regular functions $f^\ell:\R^{d_\ell}\rightarrow \R$ such that 
\begin{equation}
\label{eq:mc_approx}\lim_{\ell\rightarrow\infty} \E(f^{\ell}(X^\ell))=\E(Y)\,,
\end{equation}
where for each $\ell\ge 0$, $X^\ell$ denotes a $d_\ell$-dimensional 
standard normal vector. In most cases the $f^\ell$ will be the discrete 
versions of a function defined on the Brownian paths with $d_\ell$ 
discretization nodes, and typically $d_\ell=2^\ell$. 
A standard examples is provided by the 
fixed strike Asian option, which has payoff
\[
f(B):=\max\left(\frac{1}{T}\int_0^T S_0 \exp\left(\sigma\sqrt{T} B_{t/T}+(r-\frac{\sigma^2}{2})t\right) dt\,,\,K\right)\,,
\]
where $B$ is a standard Brownian motion, $S_0$ is the stock price at time $0$, 
$K$ is the strike of the option, $\sigma$ is the volatility and $r$ is the interest rate. 
$B$ will be approximated by a discrete path of the form $SV^\ell X^\ell$ where,
for example, $V^\ell$ is the orthogonal transform corresponding 
to $d_\ell=2^\ell$-dimensional PCA.

Eqn. (\ref{eq:mc_approx}) states that there exists
a sequence of algorithms which approximate $\E(Y)$ with increasing accuracy.
For example, if $f^\ell(X^\ell)$ has finite variance, we can approximate
$\E(Y)$ by $\frac{1}{N}\sum_{k=0}^{N-1} f^\ell(X^\ell_k)$  using sufficiently
large $\ell$ and $N$, where $(X^\ell_k)_{k\ge 0}$ is a sequence of 
independent standard normal vectors.

Usually, evaluation of $f^\ell(X^\ell_k)$ becomes more costly with increasing
$\ell$. Multilevel methods sometimes help us to save significant
proportions of computing time by computing more samples for the coarser
approximations, which need less computing time but have higher variance.

We have, for large $L$, 
\begin{equation}
\label{eq:mlmcrep}
\begin{aligned}
\E(Y)&\approx \E\left(f^L(X^L)\right)\\
&=\E\left(f^0(X^0)\right)+\sum_{\ell=1}^L \E\left(f^\ell(X^\ell)\right)-\E\left(f^{\ell-1}(X^{\ell-1})\right)\\
&=\E\left(f^0(X^0)\right)+\sum_{\ell=1}^L \E\left(f^\ell(X^\ell)\right)-\E\left(f_c^{\ell-1}(X^{\ell})\right)\\
&=\E\left(f^0(X^0)\right)+\sum_{\ell=1}^L \E\left(f^\ell(X^\ell)-f_c^{\ell-1}(X^{\ell})\right)\,,
\end{aligned}
\end{equation}
where $(f_c^\ell)_{\ell\ge 0}$ is an arbitrary sequence of functions 
$f_c^\ell:\R^{d_{\ell+1}}\rightarrow \R$ with 
$\E(f_c^{\ell-1}(X^\ell))=\E(f^\ell(X^\ell))$. The ``c'' in $f_c^\ell$ stands
for ``coarse level''. 

The most basic example for $f_c^\ell$ is given for  
$d_\ell=m^\ell$ by  $f_c^\ell=f^\ell\circ C_{m,\ell}$, where $C_{m,\ell}$
is the linear map defined by the matrix
\[
(C_{m,\ell})_{i,j}:=\left\{\begin{array}{cl}\frac{1}{\sqrt{m}}&\quad\mbox{ if }~(i-1)m+1\leq j\leq i\,m\,,\;1\le i\le m^\ell\\0& \quad\mbox{ else }\end{array}\right.\,.
\]
For example,  
\[
C_{2,\ell}:=
\left(
\begin{array}{cccccccc}
\frac{1}{\sqrt{2}}&\frac{1}{\sqrt{2}}&0&0&0&\dots&0&0\\
0&0&\frac{1}{\sqrt{2}}&\frac{1}{\sqrt{2}}&0&\dots&0&0\\
\vdots&\vdots&\vdots&\vdots&\vdots& &\vdots&\vdots\\
0&0&0&0&0&\dots&\frac{1}{\sqrt{2}}&\frac{1}{\sqrt{2}}
\end{array}
\right)\,.
\]
In general, $f_c^\ell$ is chosen in a way to get small variances for 
the $f^\ell(X^\ell)-f_c^{\ell-1}(X^{\ell})$. 

Equation (\ref{eq:mlmcrep}) becomes useful if, as is often the case in
practice, the expectation
$\E\left(f^\ell(X^\ell)-f^{\ell-1}(C_{m,\ell}X^{\ell})\right)$ 
can be approximated to the required level of accuracy using less function
evaluations $N_\ell$ for bigger $\ell$ while the costs $c_\ell$ 
per function evaluation
increases. Suppose the error  of approximation of
$\E\left(f^\ell(X^\ell)-f_c^{\ell-1}(X^{\ell})\right)$ using $N_\ell$
points is $e_\ell(N_\ell)$. We choose $N_0,\ldots,N_L$ so that 
\[
e_0(N_0)+\ldots+e_L(N_L)\le \varepsilon
\] 
while minimizing the total cost
\[
c=c_0 N_0 +\ldots+c_L N_L \,.
\]
In that way the total computation cost is typically much lower than
it would be if $\E(f^L(X^L))$ would be computed directly.

One typical situation is the numerical solution of a stochastic 
differential equation  using time discretization with
$d_\ell$ time steps and $f^\ell$ is some function on the set of solution paths.
See \cite{giles} for how to exploit this representation for Monte Carlo
simulation. See also \cite{gilwat} for the combination
of the multilevel technique with QMC.

\exclude{
In finance,
$f^\ell$ is typically of the form
$f^\ell(X)=\psi(h^\ell(X))$ for some functions
$h^\ell:\R^{d_\ell}\longrightarrow\R$  and $\psi:\R\longrightarrow\R$. In that 
context, $h^\ell$ is some function taking as its argument a discrete (geometric) Brownian path,
like the maximum or the average, and $\psi$ is the payoff that depends on the
outcome of $h^\ell$.
\begin{eqnarray}
\E\left(f^L(X^L)\right)
\nonumber&=&\E\left(f^0(X^0)\right)+\sum_{\ell=1}^L \E\left(\psi\left(h^\ell(X^\ell)\right)-\psi\left(h^{\ell-1}(C_{m,\ell}X^{\ell})\right)\right)\\
\nonumber&=&\E\left(\psi\left(h^0(X^0)\right)\right)+\sum_{\ell=1}^L \E\left(g^\ell\bigl(h^\ell_1(X^\ell),h^\ell_2(X^{\ell})\bigr)\right)\,,
\end{eqnarray}
where $h_1^\ell=h^\ell$, $h_2^\ell=h^{\ell-1}\circ C_{m,l}$ and 
$g^\ell(y_1,y_2)=\psi(y_1)-\psi(y_2)$.
}

%\subsection{(Quasi-)Monte Carlo calibration}\label{sec:calibration}

\subsection{Examples}\label{sec:examples}

Consider the problem of valuating an Asian option in the Heston model.
We solve the SDE using the simple Euler-Maruyama method eqn. \eqref{eq:euler}.
The model parameters are
$s_0=100$, $v_0=0.3$, $r=0.03$, $\rho=0.2$, $\kappa=2$, $\theta=0.3$, $\xi=0.5$, the option
parameters are $K=100$, $T=1$. 
The SDE is solved using a two-dimensional Brownian motion with 32 equally
spaced time steps. For that we need 64 independent standard normal variables
per QMC evaluation.
Since the problem is relatively high-dimensional we want to apply an orthogonal 
transform to the input variables. It is near at hand to apply one transform
for each of the two Brownian paths, but at least for this example it seems
to be better to use one 64-dimensional transform. 

We use the classical Sobol sequence for integration. We add a
64-dimen\-sion\-al random shift to the sequence and plot the $\log_2$ of the
standard deviation over $64$ integral evaluations each using $2^m$ points of
the sequence, $m=2,\ldots,10$.

The left hand graph Figure \ref{fig:heston-conv} shows the $\log_2$ of the
standard deviation along $m$ for 4 different transforms: the identity,
``Forward'', the orthogonal
transform corresponding to the Brownian bridge (i.e., the inverse Haar
transform), ``BB'', the one corresponding to PCA and the Brownian bridge
applied separately to the inputs of the two Brownian paths, ``BB2''.
On the $x$ axis we plot the $\log_2$ of the number of integration points, i.e.,
$m$, while along the $y$ axis we plot the $\log_2$ of the standard deviation
of the result over $64$ runs.

We can see that, as in many practical examples, the PCA performs best.
Maybe surprisingly the idea of using two independent Brownian bridge 
constructions performs worse than the two combined transforms, but still 
much better than the identical transform.

We complement this graph be the corresponding one for the example from
\cite{papa}. The payoff of this ``ratchet'' option is
\[
f(S_\frac{T}{d},S_\frac{2 T}{d},\ldots,S_{T})
=\frac{1}{d}\sum_{j=1}^d 1_{[0,\infty)}
\left(S_\frac{j T}{d}-S_\frac{(j-1) T}{d}\right)S_\frac{j T}{d}\,.
\]
The errors are plotted on the right hand side of 
\ref{fig:heston-conv}. We can see that the orthogonal transforms that
were so successful in the case of an Asian option now perform worse
then the identity. 

\begin{figure}[h]
\begin{center}
\includegraphics[scale=0.6]{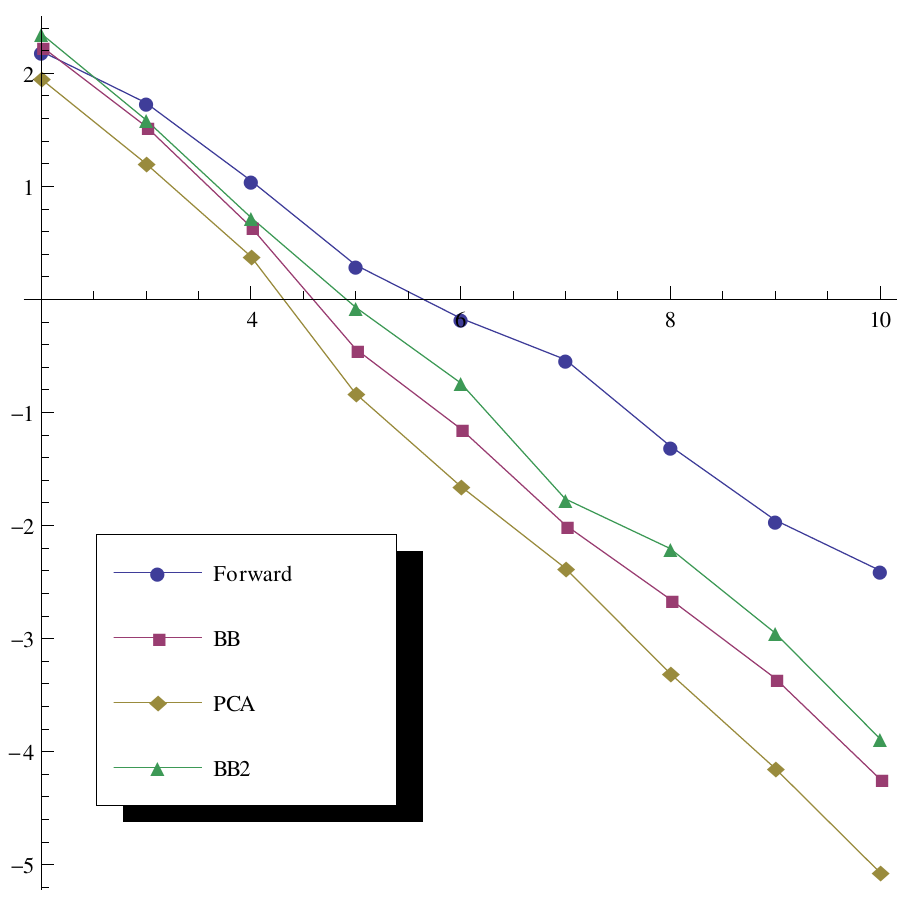}
\includegraphics[scale=0.6]{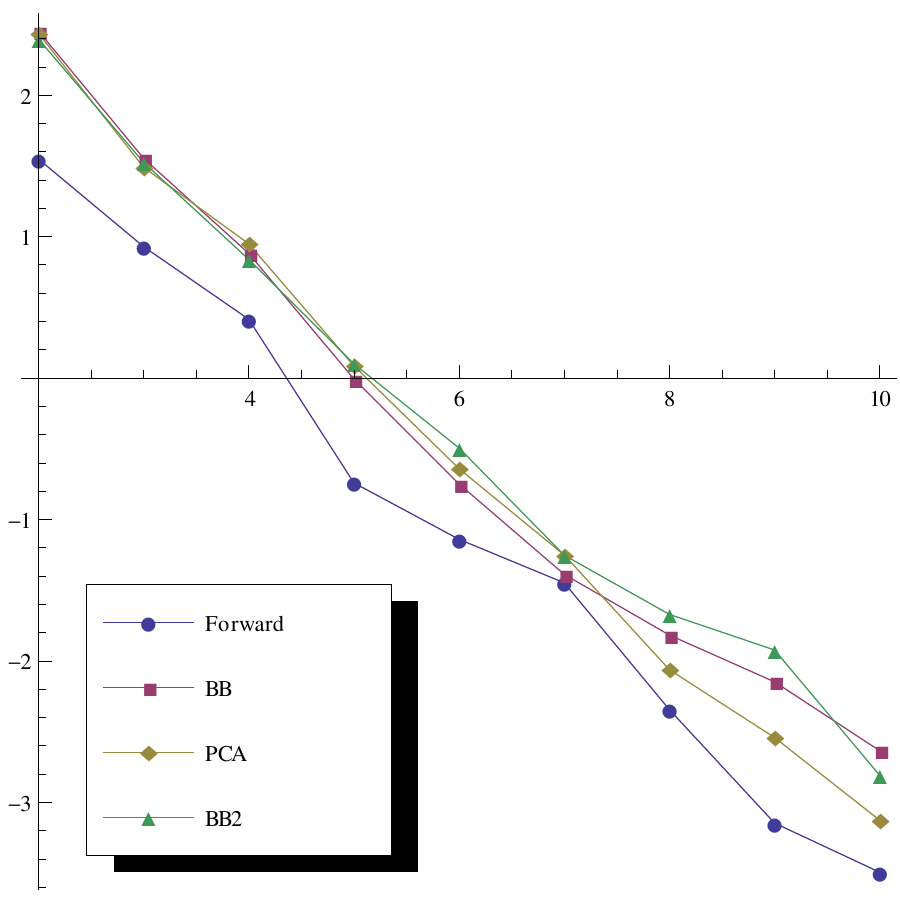}
\end{center}
\caption{Left: Convergence of the price of an Asian option under 
different transforms. Right: same graph for the ratchet option.}\label{fig:heston-conv}
\end{figure}

Thus it has to be kept in mind that the choice of the orthogonal 
transform has to be in line with the payoff function. How this should be
done exactly, and for which types of payoffs it accelerates convergence, 
is still subject to research. See for example \cite{irrgeher1,irrgeher2},
where it is tried to choose the orthogonal transform in a way that
puts as much variance as possible into  
the dependence of the first input variable. To this end the payoff is 
approximated by a linear function $g$ (``regression'') and an orthogonal
 (Householder-)transform $V$ is computed such that $g\circ V$ only depends on
$X_1$. This $V$ is taken as the orthogonal transform for the original
problem.\\

We conclude with an example in which multilevel Monte Carlo is combined with
orthogonal transforms and QMC.  We compare the multilevel QMC method
together with the regression algorithm from \cite{irrgeher1} with 
multilevel Monte Carlo and
multilevel quasi-Monte Carlo (forward and PCA sampling) numerically. For that
we choose the parameters in a Black-Scholes model as $r=0.04$, $\sigma=0.3$,
$S_0=100$, and we aim to value an Asian call option with parameters $K=100$ and
$T=1$. 
At the finest level we choose $2^{10}$ discretization
points and at each coarser level the number of points is divided in by 2, 
i.e. $L=10$ and $m=2$. The number of sample points are doubled at each level
starting with $N_L$ sample points at the finest level $L$. For the QMC
approaches we take a Sobol sequence with a random shift. In Table
\ref{tbl:asian1} we compare for different values $N_L$ both the average and the
standard deviation of the price of the Asian call option based on $1000$
independent runs. Moreover, the average computing time for one run is given in
brackets. As we can see, the regression algorithm yields the lowest standard
deviation, but the computing time of the regression algorithm is slightly 
higher than that for the forward method. The regression algorithm outperforms 
the PCA construction measured both by standard deviation and computing
time. 

\begin{table}[ht]
\hspace{-1.5cm}
%\begin{center}
\scriptsize
\begin{tabular}{l|cc|cccccc}
&	\multicolumn{2}{c|}{multilevel}    				&\multicolumn{6}{c}{multilevel QMC}\\
&	\multicolumn{2}{c|}{Monte Carlo}		&\multicolumn{2}{c}{~~forward~~}		&\multicolumn{2}{c}{~~\quad PCA\quad~~} &\multicolumn{2}{c}{~~regression~~}\\\hline
$N_L$~~ & ~average~ & ~~~stddev~~~ & ~average~ & ~~~stddev~~~ & ~average~ & ~~~stddev~~~ & ~average~ & ~~~stddev~~~ \\\hline\hline
2	& 7.717 & $0.41\times 10^{0}$	& 7.735	& $0.19\times 10^{-1}$	& 7.736	& $0.16\times 10^{-1}$	& 7.739	& $0.10\times 10^{-1}$\\
&\multicolumn{2}{c|}{(0.0057\,s)}		&\multicolumn{2}{c}{(0.0057\,s)}	&\multicolumn{2}{c}{(0.0088\,s)}	&\multicolumn{2}{c}{(0.0069\,s)}\\\hline
4	 & 7.738& $0.19\times 10^{0}$ & 7.734& $0.71\times 10^{-2}$ & 7.736& $0.44\times 10^{-2}$ & 7.738& $0.29\times 10^{-2}$\\
& \multicolumn{2}{c|}{(0.0074\,s)}&\multicolumn{2}{c}{(0.0074\,s)}&\multicolumn{2}{c}{(0.0118\,s)}&\multicolumn{2}{c}{(0.0091\,s)}\\\hline
8	 & 7.748& $0.54\times 10^{-1}$	& 7.737& $0.30\times 10^{-2}$			& 7.737& $0.14\times 10^{-2}$			& 7.736& $0.10\times 10^{-2}$\\
& \multicolumn{2}{c|}{(0.0101\,s)}	& \multicolumn{2}{c}{(0.0100\,s)}		& \multicolumn{2}{c}{(0.0165\,s)}	& \multicolumn{2}{c}{(0.0124\,s)}\\\hline
16 & 7.746& $0.40\times 10^{-1}$	& 7.736& $0.11\times 10^{-2}$& 7.737& $0.69\times 10^{-3}$& 7.736& $0.30\times 10^{-3}$\\
& \multicolumn{2}{c|}{(0.0157\,s)}	& \multicolumn{2}{c}{(0.0157\,s)} & \multicolumn{2}{c}{(0.0279\,s)}	& \multicolumn{2}{c}{(0.0194\,s)}\\\hline
32 & 7.728& $0.31\times 10^{-1}$	& 7.736& $0.49\times 10^{-3}$		& 7.737& $0.21\times 10^{-3}$	  & 7.736& $0.10\times 10^{-3}$\\
& \multicolumn{2}{c|}{(0.0266\,s)}	& \multicolumn{2}{c}{(0.0265\,s)} & \multicolumn{2}{c}{(0.0585\,s)}	& \multicolumn{2}{c}{(0.0326\,s)}\\\hline
64 & 7.739& $0.81\times 10^{-2}$	& 7.736 & $0.20\times 10^{-3}$	& 7.737	& $0.69\times 10^{-4}$	& 7.737& $0.32\times 10^{-4}$\\
& \multicolumn{2}{c|}{(0.0486\,s)}	& \multicolumn{2}{c}{(0.0484\,s)}	& \multicolumn{2}{c}{(0.1202\,s)} & \multicolumn{2}{c}{(0.0583\,s)}\\
\end{tabular}
\caption{Multilevel (Q)MC using $2^{10}$ time steps $(L=10)$. The average and the standard deviation of the option price are based on $1000$ runs. The average computing time is given in brackets.}\label{tbl:asian1}
%\end{center}
\end{table}

\exclude{
\subsection{Exercises}

\begin{enumerate}
\item Prove the general inversion method
\item Another special case of great practical interest is where $X$ takes 
values in $\Z$. Here, the distribution function is obviously not invertible.
However, usually the pseudo inverse can be reasonably well computed.

\item \label{ue:cdfinv} Let $f:\R\longrightarrow \R$ be a positive probability density function, i.e. 
\begin{itemize}
\item $f$ is measurable;
\item $f(x)>0$ for all $x\in\R$;
\item $\int_\R f(x) dx =1$.
\end{itemize}
Define $f(x):=\int_{-\infty}^xf(\xi)d\xi$. Then $F:\R\longrightarrow (0,1)$ is
bijective.

\end{enumerate}
}

%\bibliography{lit}

%\printindex

\end{document}